\documentclass[a4paper,twocolumn,english,showpacs,nofootinbib,nobibnotes,aps,prx,floatfix]{revtex4-2}
\usepackage[utf8]{inputenc}
\usepackage{amsmath}
\usepackage{amsthm}
\usepackage{amssymb}
\usepackage{graphicx}
\usepackage{hyperref}
\hypersetup{
 colorlinks=true,      
linkcolor=red,        
citecolor=blue,
filecolor=magenta, 
urlcolor=black          
}

\usepackage[caption=false]{subfig}
\usepackage{xparse}

\makeatletter

\newcounter{thmc}

\newtheorem{theorem}[thmc]{Theorem}
\newtheorem{corollary}[thmc]{Corollary}
\newtheorem{lemma}[thmc]{Lemma}

\usepackage{braket}
\usepackage{dsfont}

\makeatother

\usepackage{babel}

\global\long\def\trace{\operatorname{Tr}}
\global\long\def\diag{\operatorname{diag}}
\global\long\def\ketbra#1#2{\ket{#1}\!\bra{#2}}

\global\long\def\one{\mathds{1}}

\newcommand{\kommentar}[1]{}
\newcommand{\R}{\mathcal{R}_\mathcal{A}}
\newcommand{\Rt}{\tilde{\mathcal{R}}_\mathcal{A}}

\NewDocumentCommand\opti{smmm>{\SplitList{;}}m} {
\begingroup%
\setlength{\belowdisplayskip}{-0.6\baselineskip}%
\IfBooleanTF{#1}{%
    \begin{alignat*}{2}
        & \underset{#3}{\text{#2}} & & #4 \\
        & \text{subject to~~}
        \ProcessList{#5}{ \insertopticonst }
        & &
    \end{alignat*}%
    }{%
    \begin{alignat}{2}
        & \underset{#3}{\text{#2}} & & #4 \\
        & \text{subject to~~}
        \ProcessList{#5}{ \insertopticonst }
        & & \nonumber
    \end{alignat}%
    }%
\endgroup%
}%
\newcommand\insertopticonst[1]{& & #1\\&}

\begin{document}

\title{Probing the geometry of correlation matrices with randomized measurements}

\author{Nikolai Wyderka}
\affiliation{Institut für Theoretische Physik III, Heinrich-Heine-Universität Düsseldorf, Universitätsstr.~1, D-40225 Düsseldorf, Germany}

\author{Andreas Ketterer}
\affiliation{Fraunhofer Institute for Applied Solid State Physics IAF, Tullastr.~72, 79108 Freiburg, Germany}

\date{\today}

\begin{abstract}
The generalized Bloch decomposition of a bipartite quantum state gives rise to a correlation matrix whose singular values provide rich information about non-local properties of the state, such as the dimensionality of entanglement. While some entanglement criteria based on the singular values exist, a complete understanding of the geometry of admissible correlation matrices is lacking. We provide a deeper insight into the geometry of the singular values of the correlation matrices of limited Schmidt number. First, we provide a link to the framework of randomized measurements and show how to obtain knowledge about the singular values in this framework by constructing observables that yield the same moments as one obtains from orthogonal averages over the Bloch sphere. We then focus on the case of separable states and  characterize the boundary of the set of the first two non-vanishing moments by giving explicit constructions for some of the faces and extremal points. These constructions yield a connection between the geometry of the correlation matrices and the existence problems of maximal sets of mutually unbiased bases, as well as SIC-POVMs.
\end{abstract}
\maketitle

\section{Introduction}\label{sec:intro}

Quantum states exhibit a plethora of non-classical effects such as entanglement, coherence, and Bell non-locality which  play an integral role for the development of novel quantum technologies. For this reason, the study of fundamental properties of quantum mechanics is central theme of ongoing research even almost one hundred years after the initial formulation of quantum theory. 
Of particular interest in this respect is the non-local character of quantum entanglement which acts as a resource for applications like long-range quantum communication \cite{muralidharan2016optimal}, metrology \cite{toth2012multipartite} and various other quantum information tasks \cite{friis2019entanglement}. 

Nevertheless, the fundamental features of entanglement remain elusive. For example, for states shared between more than two parties, even a unique notion of maximal entanglement does not exist, hinting at the rich underlying structure of multiparticle entanglement. Also bipartite systems are subject of active research, given the fact that deciding whether a given state is entangled or not, is a NP-hard problem \cite{gharibian2010strong}. 
Especially in the case of bipartite quantum systems, a large toolchain of criteria exists that allow to decide (and in some cases quantify) entanglement for some states. Most prominently, the entanglement of bipartite states of joint dimension no greater than six can be efficiently decided by the PPT criterion \cite{guhne2009entanglement}. For larger-dimensional systems, however, only sufficient criteria exist.

Some of these entanglement criteria are based directly on the Bloch coefficients of the states: expanding a bipartite quantum state in terms of a local operator basis allows to represent the state as a $d^2 \times d^2$-dimensional matrix of expansion coefficients. Famous entanglement criteria like the computable crossnorm or realignment (CCNR) criterion \cite{rudolph2005further, chen2002matrix} or the de~Vicente criterion \cite{de2006separability} (or a mixture of both \cite{Sarbicki-2020-Afamilyofmultipar}) are based on norms of the Bloch matrix, or more specifically, on its singular values. Nevertheless, these criteria have at least three flaws. First of all, it is unclear whether the known criteria are the strongest in terms of those singular values. Second, they only allow to detect entanglement, but not quantify it. Finally, in order to apply them experimentally, usually a full state tomography is performed which requires a large number of well-defined measurements. 

To remedy this last flaw, a connection between the moments of randomized measurement results of certain observables and the singular values of the Bloch matrix has been established recently \cite{PhysRevLett.122.120505,Imai-2020-Boundentanglementf}. It allows to get information on the symmetric polynomials of degrees two and four of the singular values by relatively few measurements and mild requirements on the measurement apparatus. However, it was unclear whether observables exist that would yield higher-order polynomials as well, such that the singular values can be uniquely reconstructed.

In the present work we make progress in three main directions. First, we extend the known entanglement criteria that are based on the singular values such that they not only allow to detect whether a state is entangled or not, but also give lower bounds on the dimensionality of the entanglement via its Schmidt number.
Second, we expand the aforementioned connection between randomized measurements and the singular values of the Bloch matrix. We show that in any fixed local dimension $d$, local observables exist that yield as their moments the even symmetric polynomials of the eigenvalues up to degree $d$, but no more (except for dimension $d=2$ and $d=3$). This highlights an interesting connection between Bloch sphere averages in $O(d^2-1)$ and local unitary averages in $SU(d)$. 
Third, we give a partial answer to the question of whether stronger entanglement criteria exist when the singular values are known completely, by exploring in detail the landscape spanned by the first two non-vanishing moments of the constructed observables. We show that the margin of potential improvement over the de~Vicente criterion is in general very small and might even vanish.

The paper is organized as follows. In Sec.~\ref{sec:corrent} we introduce the multi-dimensional Bloch vector representation and present novel criteria for the detection of the dimensionality of two-particle entanglement. Further on, in Sec.~\ref{sec:uniorthmom}, we introduce the notion of unitary and orthogonal moments and show in Sec.~\ref{sec:secfourthmoments} how they can be applied to experimentally test the aforementioned criteria. In Sec.~\ref{sec:boundarystates} we proceed to discuss the optimality of the entanglement detection by characterizing the boundary of the set of separable states in the space of spanned by the two lowest non-vanishing moments. In particular, we highlight connections between these findings and the existence problem of mutually unbiased bases and SIC-POVMs in Sec.~\ref{sec:mubs}. In Sec.~\ref{sec:numbermeasurements}, we simulate the randomized measurement process for the family of isotropic states in order to obtain a guess on the required number of measurements to implement our new criterion experimentally. Finally, we conclude in Sec.~\ref{sec:conclusion}.

\section{Correlation matrix and entanglement dimensionality}\label{sec:corrent}

Throughout this paper, we will be concerned with bipartite quantum states of dimension $d_1 \times d_2$, and we assume w.l.o.g.~that $d_1 \leq d_2$. We decompose this state in terms of local operator bases $\{\lambda_i\}_{i=0}^{d_1^2-1}$, $\{\tilde{\lambda}_i\}_{i=0}^{d_2^2-1}$, which we assume to be orthogonal, such that $\trace(\lambda_i \lambda_j^\dagger) = d_1 \delta_{ij}$,  $\trace(\tilde{\lambda}_i \tilde{\lambda}_j^\dagger) = d_2 \delta_{ij}$ and $\lambda_0 = \one_{d_1 \times d_1}$, $\tilde{\lambda}_0 = \one_{d_2 \times d_2}$. This choice guarantees that all other elements of the bases are traceless. Note that all of the results presented here are independent of the specific choice of bases, and popular options are the hermitian Gell-Mann matrices or the unitary Weyl basis \cite{Bengtsson-2006-GeometryofQuantum}.

In this basis, an arbitrary bipartite quantum state $\rho$ can be decomposed as
\begin{align}\label{eq:blochdeco}
    \rho &= \frac1{d_1d_2}\left[\one \otimes \one + \sum_{i=1}^{d_1^2-1} \alpha_i \lambda_i \otimes \one + \sum_{j=1}^{d_2^2-1} \beta_j \one \otimes \tilde{\lambda}_j \right. \\ \nonumber
    & \phantom{=\frac1{d_1d_2}..}\left.+\sum_{i=1}^{d_1^2-1}\sum_{j=1}^{d_2^2-1}T_{ij} \lambda_i \otimes \tilde{\lambda}_j\right],
\end{align}
where we have split the sum according to the occurrence of identity matrices. In this way, the local states are given by $\rho_A = \trace_B(\rho) = \frac1{d_1}\left[\one + \sum_i \alpha_i \lambda_i\right]$, and likewise for $\rho_B$, i.e., they are completely determined by the $d_i^2-1$-dimensional vectors $\vec{\alpha} = (\alpha_1, \ldots, \alpha_{d_1^2})^\dagger $ and $\vec{\beta}$, whereas the (classical and quantum) correlations of the state are captured by the $(d_1^2-1)\times (d_2^2-1)$-dimensional correlation matrix $T$. 

When considering features of the state that are independent of specific local basis choices, such as entanglement, it is useful to notice that every local unitary rotation of the quantum state corresponds to an orthogonal rotation of the matrix $T$, i.e., $T \rightarrow O_1  T O_2^\text{T}$ \cite{Horodecki-1996-Information-theoreti, Thirring-1980-LehrbuchderMathema}.
As such, many entanglement criteria exist which are based solely on different norms of the matrix $T$ \cite{Badziag-2007-Experimentallyfrien, de2006separability, Vicente-2011-Multipartiteentangl, Imai-2020-Boundentanglementf}. A prominent example of such a criterion is the de~Vicente-criterion \cite{de2006separability}, which states that for separable states, 
\begin{align}\label{eq:dvic}
    \Vert T \Vert_\text{Tr} \leq \sqrt{(d_1 - 1)(d_2-1)},
\end{align}
where $\Vert \cdot \Vert_\text{Tr}$ denotes the trace norm, i.e., the sum of the singular values of $T$. Other examples of non-local features that can be decided on the level of the singular values of $T$ include the notion of faithful entanglement, which is linked to the usefulness of a state for a teleportation task \cite{Horodecki-1998-Generalteleportatio, Guhne-2020-Geometryoffaithful} and whether or not a state constitutes non-locality that is detectable by a CHSH-like inequality \cite{Horodecki-1995-ViolatingBellinequ, Verstraete-2001-Entanglementversus}.

Given the large evidence for the expressiveness of the singular values of $T$, it is of great interest to characterize exactly the set of admissible singular values in different subsets of quantum states. Here, we will focus on the subsets of states of certain Schmidt numbers, which give information about the dimensionality of the entanglement of the state. The Schmidt number of a pure quantum state is defined by its Schmidt rank, i.e., the number of non-vanishing Schmidt coefficients in their Schmidt decomposition, given by \cite{guhne2009entanglement}
\begin{align}
    \ket{\psi} = \sum_{j=1}^{r} \sqrt{\alpha_j} \ket{a_j} \otimes \ket{b_j},
\end{align}
where $\alpha_j>0$, $\sum_{j=1}^r \alpha_j = 1$, and $r$ denotes the Schmidt rank $\text{SR}(\ket{\psi})$. For mixed states, the Schmidt number (SN) is defined via \cite{terhal2000schmidt}
\begin{align}
    \text{SN}(\rho) = \min\{r\,:\,\rho = \sum_i p_i \ketbra{\psi_i}{\psi_i}\text{ where SR}(\ket{\psi_i}) \leq r\}.
\end{align}
The $p_i$ are non-negative numbers summing to 1. 
The Schmidt number lies between 1 and $\min(d_1,d_2) = d_1$. A special subset are states of Schmidt number 1, which are called separable. They can be decomposed into convex combinations of product states, i.e.
\begin{align}
    \rho_\text{sep} = \sum_{i} p_i \rho_A^{(i)} \otimes \rho_B^{(i)},
\end{align}
where the $p_i$ constitute a probability distribution.

In order to obtain a criterion for Schmidt numbers based on the singular values, we extend the already introduced de~Vicente-criterion from Eq.~(\ref{eq:dvic}):
\begin{lemma}\label{lem:dVschmidt}
Let $\rho_k$ be a bipartite quantum state as before with Bloch decomposition as in Eq.~(\ref{eq:blochdeco}) and $\text{SN}(\rho_k) \leq k$. Then
\begin{align}
    \Vert T \Vert_{\text{Tr}} \leq \sqrt{(d_1-1)(d_2-1)} + \sqrt{d_1d_2}(k-1).
\end{align}
\end{lemma}
The proof is found in Appendix \ref{app:dVschmidt}, where also a generalization for the family of criteria derived in Ref.~\cite{Sarbicki-2020-Afamilyofmultipar} for Schmidt number detection is also derived, which includes a bound for the CCNR criterion.
Note that if $d_1 = d_2 = d$, this bound simplifies to $\Vert T \Vert_{\text{Tr}} \leq kd-1$. As a remark, we emphasize that indirect bounds  on the Schmidt number using the de~Vicente criterion can be obtained via bounds on the concurrence derived in Ref~\cite{de2007lower}. These are, however, strictly weaker than the bound from Lemma~\ref{lem:dVschmidt}. Furthermore, the companion paper Ref.~\cite{liu2022characterizing} reports the same result as our Lemma~\ref{lem:dVschmidt}. We would like to stress, however, that their proof is based on bounds for the covariance criterion reported in Ref.~\cite{liu2022bounding}. Our proof is independent of this criterion and can be extended to the family of criteria of Ref.~\cite{Sarbicki-2020-Afamilyofmultipar}.

Lemma~\ref{lem:dVschmidt} allows for detection of lower bounds of Schmidt numbers in the following way: If the criterion is violated for a specific choice of $k$, then the state in question must have a Schmidt number of at least $k+1$.

A complete characterization of the set of occurring singular values in the set of states of certain Schmidt number would yield an optimal criterion for the entanglement dimensionality  whenever the singular values are known. Of course, such a criterion would only be useful if the singular values can be obtained experimentally which is a problem that we address in the following.

\section{Unitary and orthogonal moments}\label{sec:uniorthmom}

Studying singular values of the correlation matrix is also motivated by the framework of randomized measurements \cite{Emerson-2005-Scalablenoiseestim, Knill-2008-Randomizedbenchmark, Flammia-2011-DirectFidelityEsti, Aaronson-2007-Thelearnabilityof, Huang-2020-Predictingmanyprop, Morris-2019-SelectiveQuantumSt, Brandao-2020-Fastandrobustquan, Liang-2010-NonclassicalCorrela, Wallman-2012-Observerscanalways, Shadbolt-2012-Guaranteedviolation}. In many of these schemes, one studies moments of the outcome distribution of randomly locally rotated states, i.e., one obtains the quantities
\begin{align}\label{eq:rt}
    \R^{(t)}(\rho) := \intop \text{d}U_A\text{d}U_B [\trace(U_A\otimes U_B \rho U_A^\dagger \otimes U_B^\dagger \mathcal{A}\otimes \mathcal{A})]^t,
\end{align}
where $\mathcal{A}$ is some local observable and the integration spans over the local unitary groups of dimension $d_1$ and $d_2$ with respect to the normalized Haar measures $dU_A$ and $dU_B$, respectively. For now, we set $d_1=d_2\equiv d$ for simplicity.

However, in order to obtain symmetric polynomials of the singular values of $T$, one would rather be interested in moments resulting from random Bloch sphere rotations:
\begin{align}\label{eq:orthomoments}
    \mathcal{S}^{(t)}(\rho) := \frac{1}{V^2} \intop \text{d}\vec{\alpha} \text{d}\vec{\beta} [\trace(\rho \vec{\alpha}\cdot \vec{\lambda} \otimes \vec{\beta}\cdot \vec{\lambda})]^t,
\end{align}
where here the integration spans over $d^2-1$-dimensional unit vectors, the surface area of which is given by $V=2\sqrt{\pi}^{d^2-1}/\Gamma[(d^2-1)/2]$. Here, $\Gamma(z)$ denotes Euler's Gamma function with $\Gamma(n) = (n-1)!$ for positive integer $n$. In contrast to the unitary moments~(\ref{eq:rt}) defined above, we call Eq.~(\ref{eq:orthomoments}) orthogonal moments, as they correspond to averaging over orthogonal transformations $\vec{\alpha} \rightarrow O_1\vec{\alpha}$, $\vec{\beta} \rightarrow O_2\vec{\beta}$ and $T \rightarrow O_1TO_2^T$. Indeed, Eq.~(\ref{eq:orthomoments}) can be evaluated analytically to be zero for each odd $t$, and for $t=2$ and $t=4$, it yields \cite{Imai-2020-Boundentanglementf}
\begin{align}
    \mathcal{S}^{(2)}(\rho) &= \frac{1}{(d_1^2-1)(d_2^2-1)} \trace(TT^\dagger ), \\
    \mathcal{S}^{(4)}(\rho) &= \frac{3}{(d_1^4-1)(d_2^4-1)}\left[2\trace(TT^\dagger TT^\dagger ) + \trace(TT^\dagger )^2\right]. 
\end{align}
Note that $\trace(TT^\dagger) = \sum_{i=1}^{d_1^2-1} \sigma_i^2$ and $\trace(TT^\dagger TT^\dagger) = \sum_{i=1}^{d_1^2-1} \sigma_i^4$ can be written in terms of the singular values $\sigma_1, \ldots, \sigma_{d_1^2-1}$ of $T$.
A similar evaluation is possible for general $t$, and knowledge of the \emph{orthogonal} moments up to order $t=2d$ allows to recover the singular values of $T$ completely. However, the orthogonal moments are not  directly accessible in an experiment. Thus, the question arises whether it is possible to choose the observable $\mathcal{A}$ such that the unitary moments in Eq.~(\ref{eq:rt}) and the orthogonal moments coincide for all values of $t$.

The answer to this question depends strongly on the dimension. In Appendix~\ref{app:suso}, we show the following.
\begin{theorem}
Let $d\geq 2$ denote the local dimension of an $n$-partite quantum system. Then there exist unique diagonal operators $\mathcal{A}_d$, such that for all $n$-partite quantum states $\rho$
\begin{itemize}
    \item $\mathcal{R}^{(t)}_{\mathcal{A}_d}(\rho) = \mathcal{S}^{(t)}(\rho) = 0$ for all odd $t=1,3,\ldots $,
    \item $\mathcal{R}^{(t)}_{\mathcal{A}_d}(\rho) = \mathcal{S}^{(t)}(\rho)$ for the subset of even $t=0,2,\ldots,d$.
\end{itemize}
If the dimension is small, we have additionally:
\begin{itemize}
    \item If $d=3$, also the fourth moment coincides,
    \item if $d=2$, all moments coincide.
\end{itemize}
\end{theorem}
The proof as well as some examples can be found in Appendix~\ref{app:suso}.
Interestingly, for $d \geq 8$, complex diagonal entries of $\mathcal{A}_d$ are required, making it formally not an observable. However, it can still be measured by assigning complex outcomes to a projective measurement result.
Note further that the constructed observables can easily be checked to not yield the correct moments for $t>d$ (for $d>4$). Thus, these results exclude the existence of observables that recover all moments (except for $d=2$).  We also show that it applies to cases of more than two parties as well. 

Finally, we would like to stress that it is possible to explicitly construct low-rank operators that yield matching moments $t=0,1,2,3,4$ for all dimensions $d$. These are given by
\begin{align}\label{eq:Ad4}
    \mathcal{A}_d^{(4)} = \text{diag}( \pm \kappa_1, \pm \kappa_2, 0, \ldots, 0)
\end{align}
with 
\begin{align}
    \kappa_{1,2} = \frac12 \sqrt{d \pm \sqrt{d\left(8-d-\frac{20}{d^2+1}\right)}},
\end{align}
and play a central role in the remainder of the manuscript (see  App.~\ref{app:suso} for further details).

\begin{figure*}[t]
    \centering
    \includegraphics[width=1.0\columnwidth]{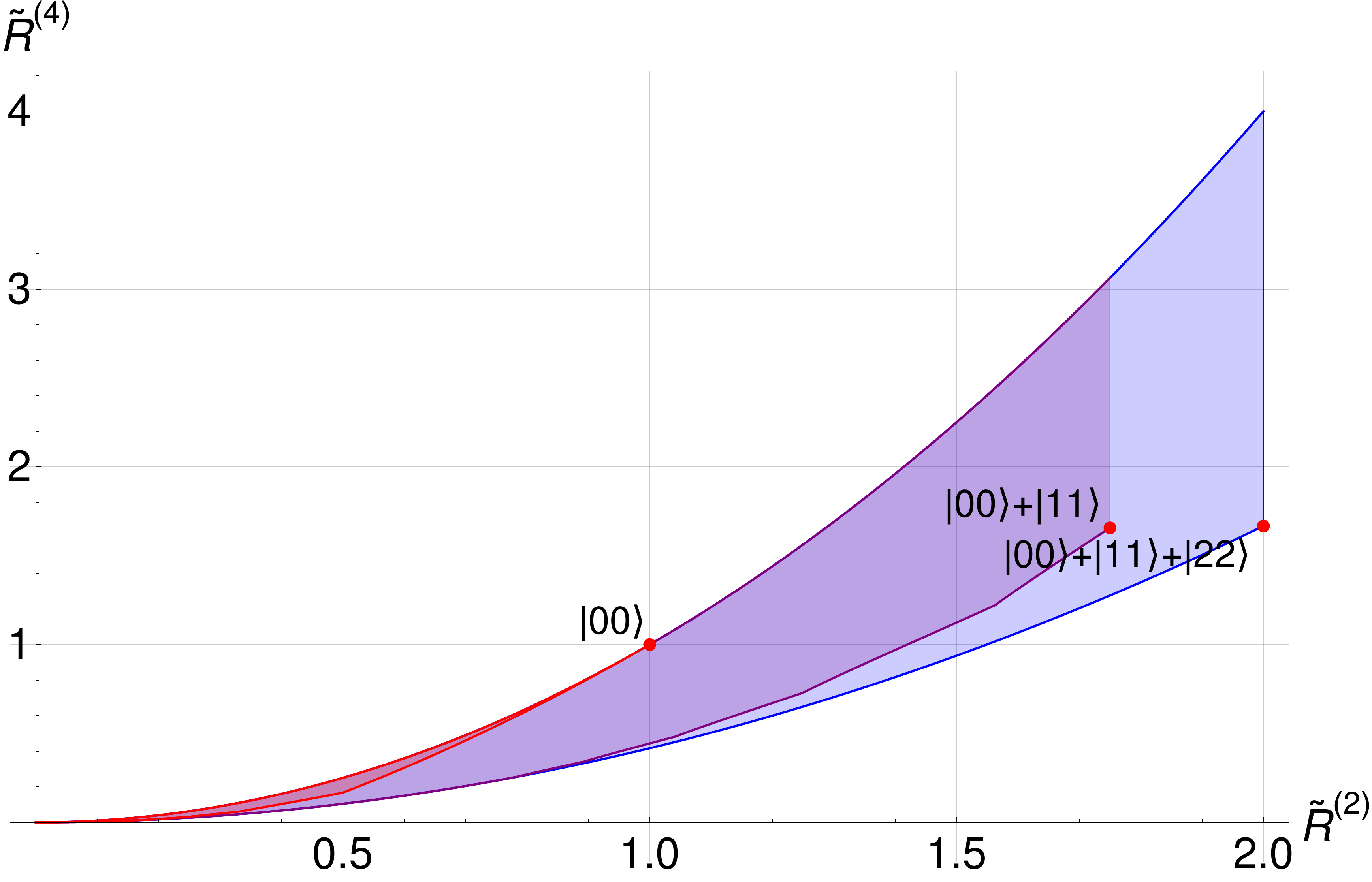}
    \includegraphics[width=1.0\columnwidth]{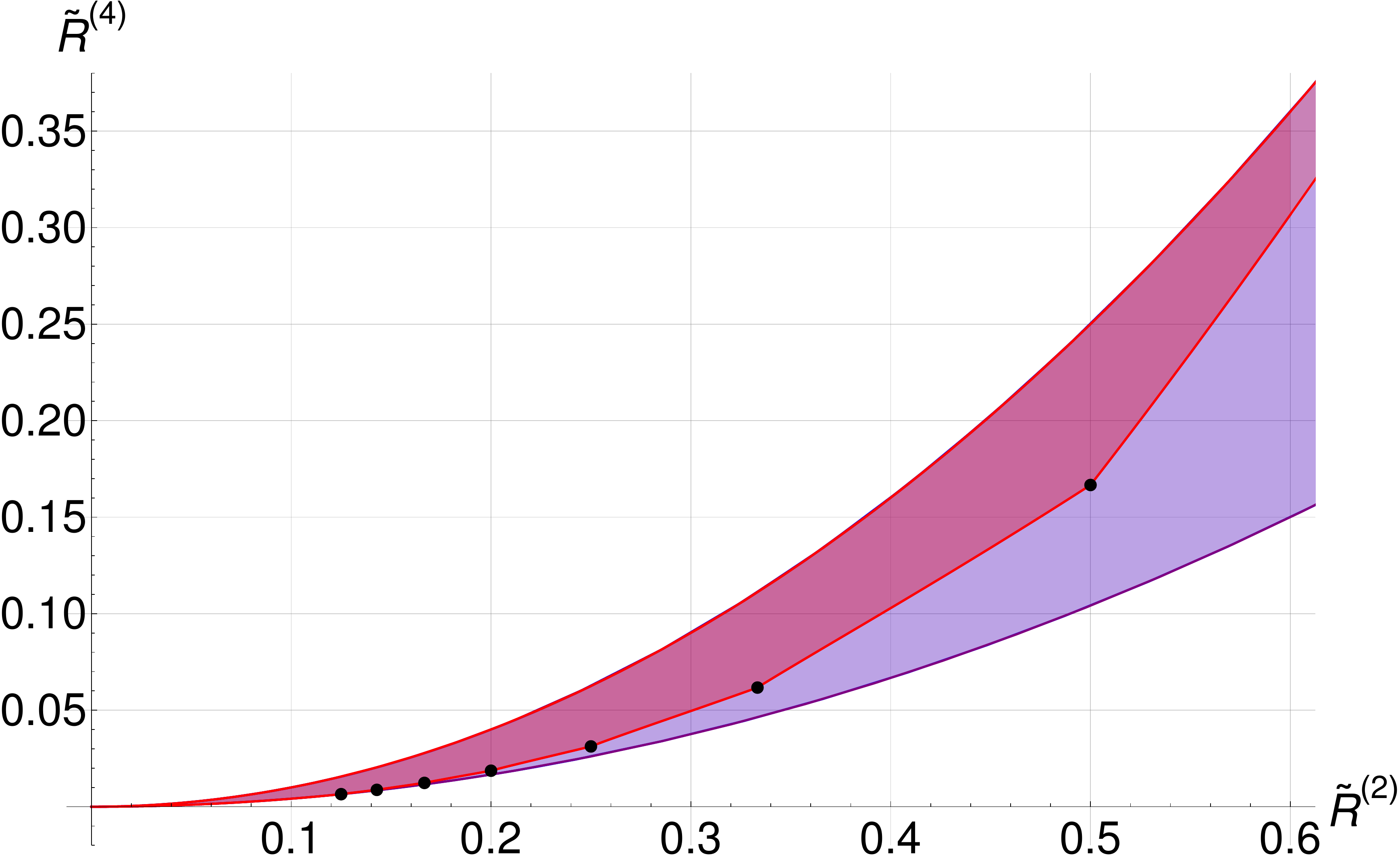}
    \caption{Left: The space of all states (red, violet and blue) and the subspaces of states of Schmidt number 2 (violet) and separable states (red) according to the generalized de~Vicente criterion in the space spanned by the second and fourth moment given by Eq.~(\ref{eq:rt}) for $d_1 = d_2 =3$.
    Right: zoomed in lower left part of the plot. The lower boundary of the separable set is given by Eq.~(\ref{eq:flb}) and exhibits $d^2-2$ kinks (displayed as black dots), which correspond to states with a dichotomous singular value structure.}
    \label{fig:fig1}
\end{figure*}

\section{Space of second and fourth moments}\label{sec:secfourthmoments}

Given the existence of these operators for moments two and four, we aim to obtain the boundaries of the sets of states of different Schmidt number in the space spanned by $\R^{(2)}$ and $\R^{(4)}$. To that end, we minimize and maximize $\R^{(4)}$ subject to keeping $\R^{(2)}$ constant and some positivity and Schmidt number constraints, respectively, in terms of the singular values. Of course, we do not know the exact constraints, but we will obtain outer approximations to the projected spaces by some necessary conditions, and then try to construct states for each point within the obtained shape to show that in some cases it is actually not an approximation, but a faithful representation of the space.

As a result we obtain the regions shown in Fig.~\ref{fig:fig1}, with $d_1 = d_2 = 3$, for the set of states of Schmidt number one (i.e., the set of separable states), two and three (i.e., all states). All in all, the respective regions represent the hierarchical entanglement structure of the space of bipartite quantum states with increasing entanglement dimensionality. In particular, Fig.~\ref{fig:fig1} visualizes the information about the entanglement structure that is gained by combining moments of different orders, as knowledge of the second moment alone detects the entanglement of strictly fewer states. Also, note that we normalize the values of the moments such that $\Rt^{(t)}(\ket{00}) = 1$ for product states, i.e.,
\begin{align}
    \Rt^{(2)} &:= (d_1+1)(d_2+1) \R^{(2)},\\
    \Rt^{(4)} &:= \frac{(d_1+1)(d_2+1)(d_1^2+1)(d_2^2+1)}{9(d_1-1)(d_2-1)} \R^{(4)}.
\end{align}
With this normalization the shape of the set of separable states only depends on the smaller of the two dimensions and exhibits $d_1^2 - 2$ kinks along the lower boundary, which can be stated analytically. In case of the sets of SN $k=2$ and $k=3$, we note that the upper bound presented in Fig.~\ref{fig:fig1} is likely not tight as its derivation is based on the rather weak purity bound  of $\trace(\rho^2) \leq 1$, and we were unable to find states lying on this boundary. The detailed optimization process, as well as the analytical form of the boundaries for all dimensions for the case $k=1$, can be found in Appendix~\ref{app:opti}.

Let us highlight some important observations for the set of separable states.
\begin{itemize}
    \item The upper bound of the red, separable set coincides with the upper bound for all states.
    \item The lower bound of the separable set is a step-wise defined function (Eq.~(\ref{eq:flb})) with kinks at the positions $\Rt^{(2)} = \frac{1}{d_1^2-1}, \frac{1}{d_1^2-2}, \ldots, \frac{1}{3}, \frac{1}{2}$. Thus, switching from $d_1$ to $d_1+1$ changes only some details on the left-hand side of the plot, leaving most of the figure unchanged. Due to our normalization, the difference of the figures for $d_1$ and $d_1+1$ vanishes for large $d_1$. As the normalization is dimension dependent, this does not imply that the criterion will be less efficient for larger $d_1$.
    \item States on the lower bound of the separable set are characterized by the fact that the singular values of the correlation matrix are of the form $(0,\ldots,0,l,m,\ldots,m)$ with $l \leq m$, and all of them summing up to $\sqrt{(d_1-1)(d_2-1)}$.
    \item States at the kinks are characterized by singular values being $(0,\ldots,0,m,\ldots,m)$, summing up to $\sqrt{(d_1-1)(d_2-1)}$.
    \item There exist entangled states with $1 < \Rt^{(2)} \leq \frac{(d_1+1)d_2}{d_1(d_2-1)}$. The maximal value of $\Rt^{(2)}$ is achieved by the maximally entangled state $\ket{\phi^+} = \frac{1}{\sqrt{d_1}}\sum_{i=0}^{d_1-1} \ket{ii}$.
    \item The lower bound of the set of all states is traced by the isotropic states, $\rho_\textrm{iso} = (1-p)\ket{\phi^+}\!\bra{\phi^+} + \frac{p}{d_1d_2}\one$. The lower bound coincides with that of the separable states for $\Rt^{(2)} \leq \frac{1}{d_1^2-1}$, which is exactly the region where $\rho_\textrm{iso}$ is separable. Thus, the whole range of entanglement of the isotropic state can be detected.
\end{itemize}

\section{Optimality of the criterion}\label{sec:boundarystates}

\begin{figure*}[t]
    \centering
    \includegraphics[width=1.82\columnwidth]{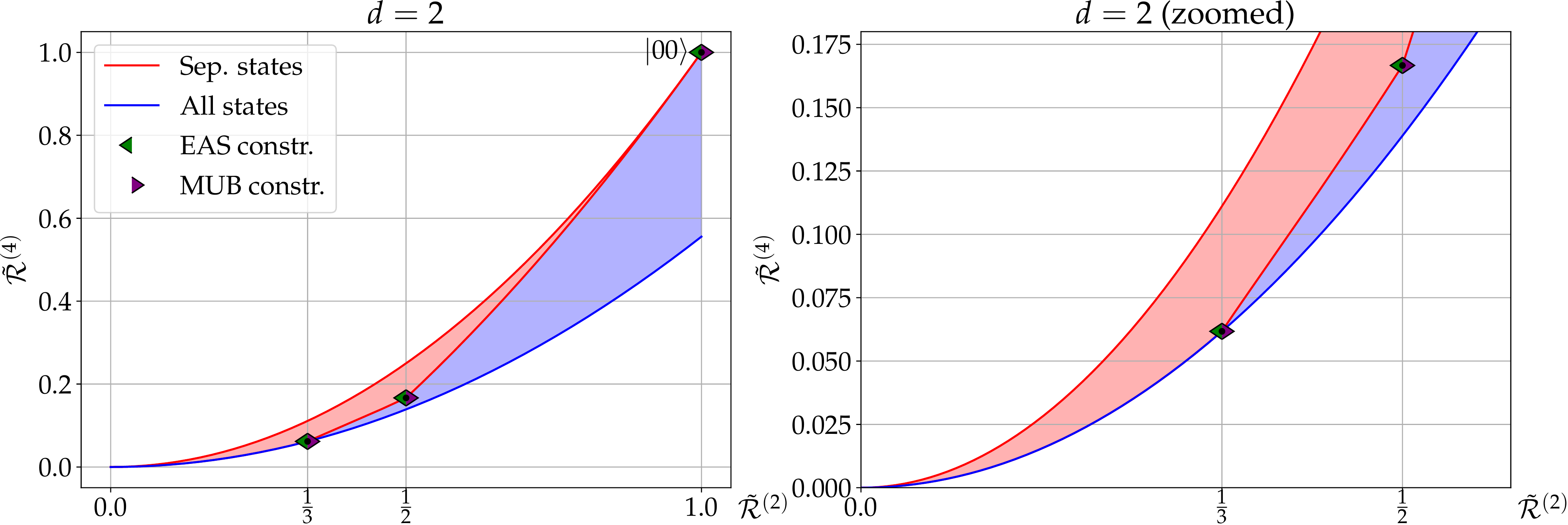}
    \includegraphics[width=1.82\columnwidth]{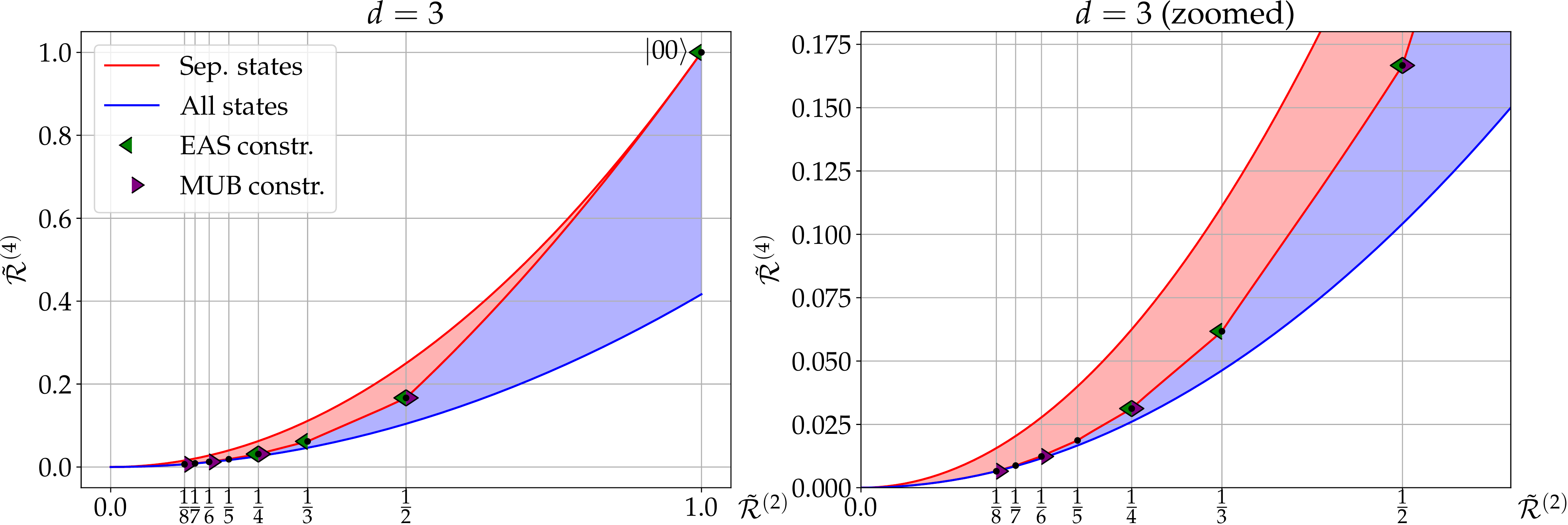}
    \includegraphics[width=1.82\columnwidth]{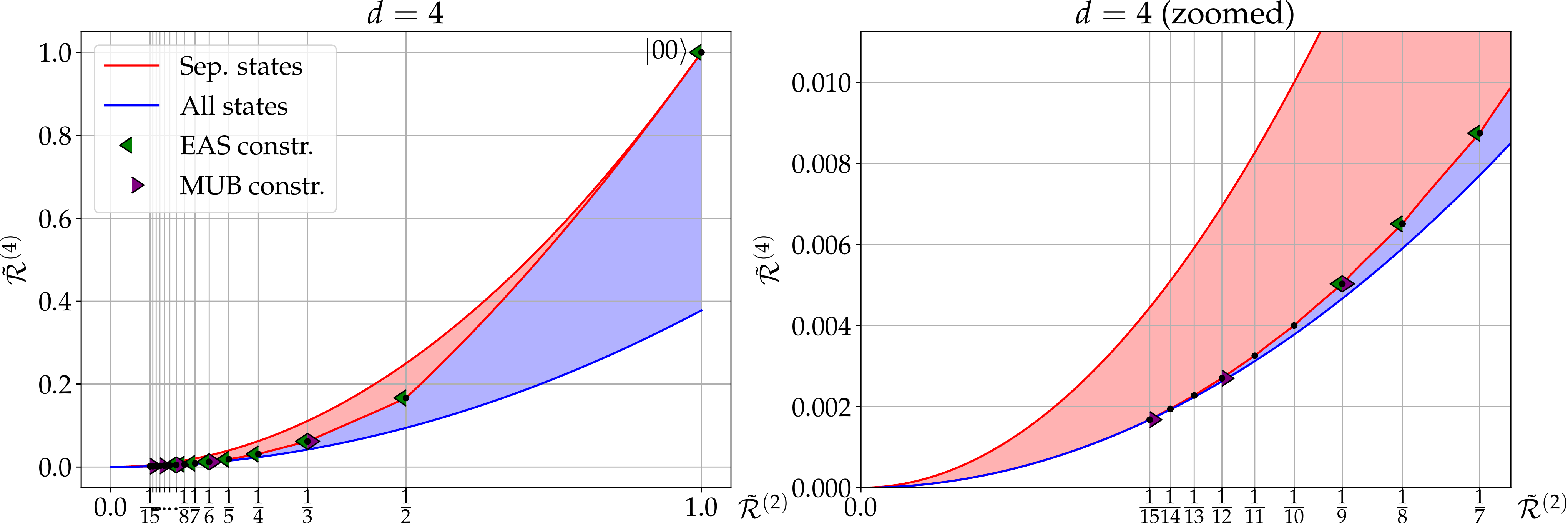}
    \includegraphics[width=1.82\columnwidth]{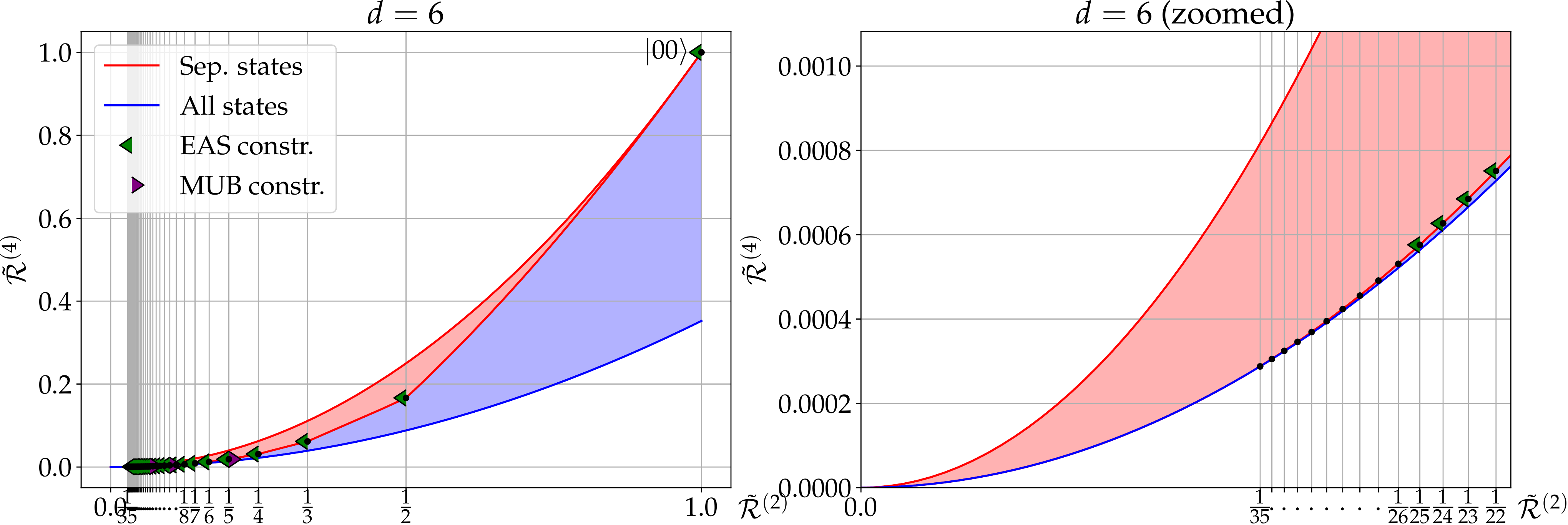}
    \caption{The moment landscape of the second and fourth moment. Displayed is the set of all states (red and blue) and the subset of separable states (red) together with its kinks on the lower boundary for $d=2,3,4$ and $6$. The markers labeled MUB (mutually unbiased basis) and EAS (equiangular set) indicate that a separable state at the corresponding point can be constructed by one of our methods.}
        \label{fig:const}
\end{figure*}

Given that we used a specific entanglement criterion in order to construct the set of separable states in the $\Rt^{(2)}$-$\Rt^{(4)}$-plane, the question remains of how tight this criterion really is. Thus, we attempt to construct a separable state for each point within the obtained figure. If this turns out to be possible, then this proves that the figure is optimal, i.e., there exists no better, smaller region, and thus, that the de~Vicente-criterion is indeed optimal when given the knowledge of $\sum_i \sigma_i^2$ and $\sum_i \sigma_i^4$.

Note that it suffices to construct separable states on the boundary of the set. All points in the interior can then be reached by mixing the boundary states with white noise in the way detailed below for the states on the upper boundary.

We start by explicitly constructing states for the upper boundary, connecting the pure product states at $(\Rt^{(2)}, \Rt^{(4)}) = (1,1)$ with the origin, where the maximally mixed state is situated. Considering the family of states
\begin{align}
    \rho_\text{upper}(p) = (1-p)\ketbra{00}{00} + \frac{p}{d_1d_2}\one, 
\end{align}
we note that the correlation matrix of $\rho_\text{upper}(p)$ is just a multiple of that of the pure product state, and therefore $\Rt^{(t)}[\rho_\text{upper}(p)] = (1-p)^t$, tracing exactly the upper boundary of the set.

For the lower boundary, we will restrict our attention to the case of $d_1 = d_2 \equiv d$.

As noted before, the lower bound for all states is traced by the isotropic state, which becomes separable for $\Rt^{(2)}\leq \frac{1}{d^2-1}$, where the lower bound of both sets coincide. Thus, we already have separable states for the lower bound in that region.
The situation for the rest of the lower bound is more complicated, but simplified by the fact that states at the kinks have a special singular value structure, namely that all of the non-vanishing singular values coincide and sum to $d-1$. 

\subsection{Constructions using equiangular sets}
An equiangular set (sometimes called equiangular lines) of size $N$ is a set of $N$ normalized states $\{\psi_i\}_{i=1}^{N}$, such that $\vert\braket{\psi_i|\psi_j}\vert = C$ for some constant $C$ if $i\neq j$ \cite{Welch-1974-Lowerboundsonthe, Tropp-2005-Complexequiangular}. A popular example is that of a SIC-POVM, where $C=\frac{1}{\sqrt{d+1}}$. Here, however, we consider sets with $C=\frac{1}{\sqrt{d}}$, from which states at the kinks can be constructed:
\begin{theorem} \label{thm:Thm1}
Let $\{\psi_i\}_{i=1}^{N}$ be an equiangular set of normalized vectors in $\mathbb{C}^d$, with $\vert\braket{\psi_i|\psi_j}\vert = \frac{1}{\sqrt{d}}$ if $i\neq j$. Then
\begin{align}
    \rho = \frac1N\sum_{i=1}^N \ket{\psi_i}\!\bra{\psi_i} \otimes \ket{\psi_i}\!\bra{\psi_i}
\end{align}
lies on the kink at $\Rt^{(2)}(\rho) = \frac{1}{N}$.
\end{theorem}
The proof can be found in Appendix~\ref{app:proofthm1}.

It follows from this theorem that, if there exist $d^2-1$ states with distance $\frac{1}{\sqrt{d}}$, we would have a symmetric separable state at each kink. However, the maximal size of such sets is subject to bounds and no systematic constructions are known for arbitrary $d$ \cite{Tropp-2005-Complexequiangular}. For small $d$, however, one can construct them explicitly.

For $d=2$, the vectors $\ket{0}, \ket{+}$ and $\ket{+_i}$ form a maximal equiangular set with the correct angle. As $d^2-1 = 3$ in this case, this yields states for all of the kinks.

For $d=3$, one can find a maximal set consisting of $4$ vectors, which is less than $8$, the number required to obtain a state for every kink, implying that one needs a different strategy for the rest of the kinks.

In order to obtain a lower bound on the maximal size of equiangular sets with overlap $\frac{1}{\sqrt{d}}$, we can borrow some results from the construction of SIC-POVMs. There, one searches for equiangular sets with $C=\frac{1}{\sqrt{d+1}}$ and it was conjectured by Zauner that the maximal set size is always given by $d^2$ \cite{ZAUNER-2011-QUANTUMDESIGNSFOU} (concrete SIC-POVMs are known for all dimensions $d\leq 193$ \cite{grasslSIC}). In order to construct sets with overlap $\frac{1}{\sqrt{d}}$, we can take a SIC-POVM of dimension $d-1$ and pad an extra dimension to the vectors. Of course, this procedure might be wasteful, such that it only yields a lower bound, and it relies on Zauner's conjecture being true. Nevertheless, it implies that we can always find at least $(d-1)^2$ vectors of the $d^2-1$ that we would need. 

Note that with this construction, we can also construct states on the boundary between the kinks:

\begin{corollary}\label{cor:cor1}
Let $\{\psi_i\}_{i=1}^{N} \cup \{\ket{\psi}\}$ be an equiangular set of size $N+1$ of normalized vectors in $\mathbb{C}^d$ with overlap $\frac{1}{\sqrt{d}}$. Then the family of states
\begin{align}
    \rho(p) = \frac{p}N\sum_{i=1}^N \ket{\psi_i}\!\bra{\psi_i} \otimes \ket{\psi_i}\!\bra{\psi_i} + (1-p) \ket{\psi}\!\bra{\psi} \otimes \ket{\psi}\!\bra{\psi}
\end{align}
traces the boundary of the separable set between the kinks at $\Rt^{(2)}[\rho(N/(N+1))] = \frac{1}{N+1}$ and $\Rt^{(2)}[\rho(1)] = \frac{1}{N}$.
\end{corollary}
The proof is given in Appendix~\ref{app:proofcor1}.

Note that in the limit of large dimension, an increasing percentage of the lower boundary is covered by this construction due to the fact, that we find states for each point between $\Rt^{(2)} = \frac1{(d-1)^2}$ and $\Rt^{(2)} = 1$ (under the assumption that Zauner's conjecture is true).

The question remains whether the missing kinks can also be covered by separable states. While we cannot answer this question completely, we can find additional states using MUBs.
\subsection{Constructions using MUBs}

A collection of orthonormal bases, $\{\ket{\psi_i^1}\}_{i=0}^{d-1}$, $\{\ket{\psi_i^2}\}_{i=0}^{d-1},\ldots$ is called mutually unbiased, if $\vert\braket{\psi_i^k|\psi_j^l}\vert = \frac{1}{\sqrt{d}}$ if $k \neq l$. For a fixed dimension $d$, there are at most $d+1$ mutually unbiased bases. However, only for prime power dimensions $d=p^m$, constructions of $d+1$ bases are known (see \cite{Bengtsson-2006-Threewaystolooka} and references therein). For $d=6$, it is widely believed that only $3$ MUBs exist. 
Indeed, taking a single vector out of each basis yields an equiangular set of the right overlap for Theorem~\ref{thm:Thm1}. However, as the maximal size of sets created in this way is given by $d+1$, this construction is not good enough for our purposes. However, a set of $m$ MUBs can still be used to construct states at some intermediate kinks at $\Rt^{(2)} = \frac1{l(d-1)}$ for $l=1,\ldots,m$:

\begin{theorem}\label{thm:Thm2}
Consider a set of $m$ MUBs  $\{\ket{\psi_j^1}\}_{j=1}^d, \ldots, \{\ket{\psi_j^m}\}_{j=1}^d$. Then the state
\begin{align}
    \rho &:= \frac{1}{md}\sum_{l=1}^m\sum_{k=1}^{d}  \ket{\psi_k^l}\!\bra{\psi_k^l} \otimes  \ket{\psi_k^l}\!\bra{\psi_k^l}
\end{align}
lies at the kink at $\Rt^{(2)} = \frac1{m(d-1)}$.
\end{theorem}
The proof is given in Appendix~\ref{app:proofthm2}.

With all of these constructions, we are able to construct states for many of the kinks, but not all of them. For example, for $d=3$, using a maximal equiangular set of size 4 , together with the MUB constructions, yields states for the kinks at $K=1,2,3,4,6$ and $8$, thus, the states at $K=5$ and $K=7$ are missing. For $d=4$, equiangular sets yield states for $K=1,\ldots,9$, and the MUB construction adds the cases $K=12$ and $K=15$, leaving open the cases $K=10,11,13$ and $14$. For $d=6$ the situation is more complex, as only sets of three MUBs are known. 
The situation is displayed in Fig.~\ref{fig:const}.
 
\section{Existence of MUBs and SIC-POVMs}\label{sec:mubs}

It is conjectured that there exist at most $3$ MUBs in $d=6$. However, there is no proof for this, in fact, there is not even a proof that the number of MUBs is not maximal, i.e., $d+1=7$.  We will focus here on the question whether $6$ or $7$ MUBs in $d=6$ exist.

According to Theorem~\ref{thm:Thm2}, we have:
\begin{itemize}
    \item If $6$ MUBs exist, then there exists a symmetric and homogeneous (meaning $p_i=p_j$) separable state at $\Rt^{(2)} =  \frac1{30}$, with a correlation matrix being proportional to a projector.
    \item If $7$ MUBs exist, then there exists a symmetric and homogeneous decomposition of the isotropic state at $\Rt^{(2)} = \frac1{d^2-1} = \frac1{35}$ with a correlation matrix being proportional to a projector.
\end{itemize}
The contraposition states that, if such symmetric, homogeneous separable states do not exist at the given coordinates, then such sets of MUBs do not exist. We do not consider the case of $m=5$ here, because an equiangular set construction exists at $\Rt^{(2)} = \frac1{25}$, making it unlikely that our trick works.

We can formulate a relaxed version of the existence problem in terms of a convex feasibility problem with rank constraint. For that, we assume w.l.o.g.~that the matrix basis we choose is hermitian, making the correlation matrix $T$ real-valued. Then, the feasibility problem can be stated like this:

\opti{find}{}{T\in \mathbb{R}^{(d^2-1) \label{eq:feasibility}\times (d^2-1)}}{T\text{ symmetric}, T\geq0,\nonumber;\one \otimes \one + \sum_{ij} T_{ij}\lambda_i \otimes \lambda_j \in \text{SEP},\nonumber;\trace(T^k)=\frac{d-1}{m^{k-1}}\quad\text{for }k\geq1.\nonumber}

Here, $\text{SEP}$ stands for the set of separable states. If the problem is infeasible, then no $m$ MUBs in dimension $d$ exist. Note that the constraints in the feasibility problem are hard to implement, but could be tackled by a combination of two SDP hierarchies \cite{Navascues-2008-Aconvergenthierarc, Yu-2020-Quantum-inspiredhie}. Using the dual representation of this problem, it might be possible to obtain infeasibility certificates that allow to analytically show the infeasibility of (\ref{eq:feasibility}) \cite{vandenberghe1996semidefinite}.

\section{Number of required measurements for entanglement detection}\label{sec:numbermeasurements}

\begin{figure*}[t]
    \centering
    \includegraphics[width=2.0\columnwidth]{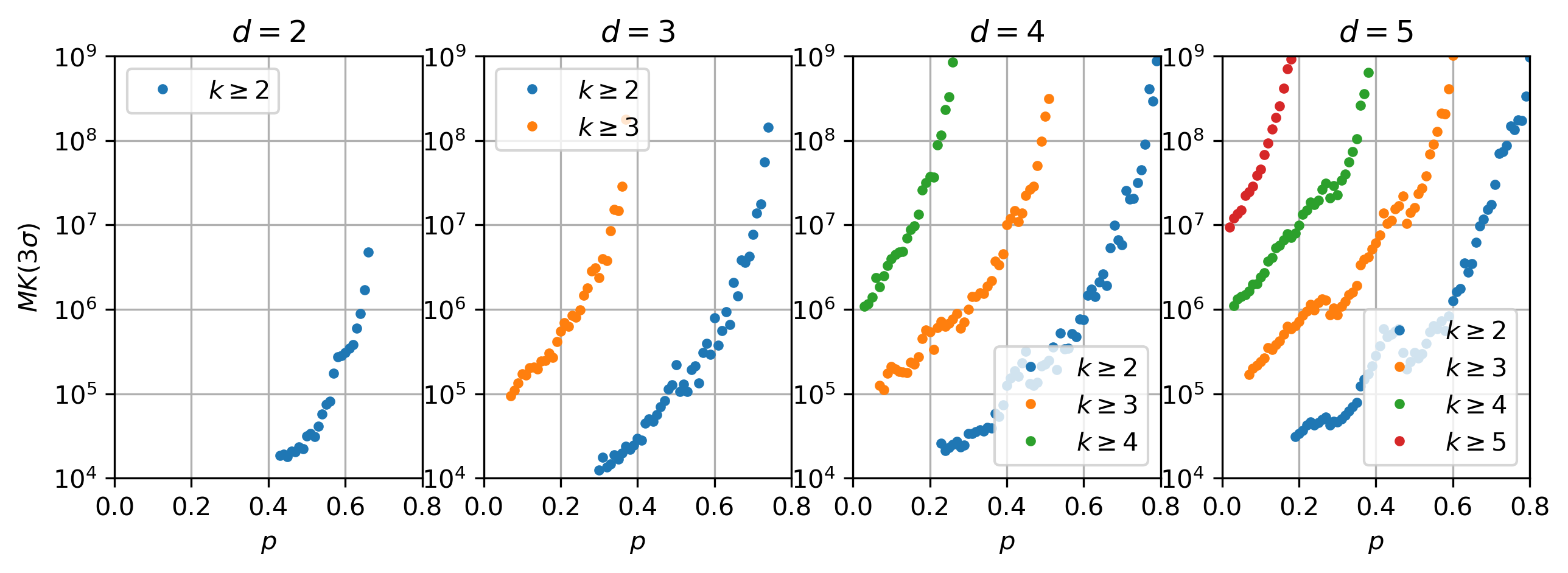}
    \caption{Estimated number of total measurements $MK$ required to observe a $3\sigma$ violation of the criterion in Lemma~\ref{lem:dVschmidt} using randomized measurements for dimensions $d=2,3,4,5$ for the isotropic state in Eq.~(\ref{eq:iso}) as a function of the noise parameter $p$. For each dimension $d$, we display the required number of measurements to detect Schmidt numbers of at least $2, \ldots, d$. Note that the curves start at some value $p_\text{min}$ where our criterion starts to be applicable (for smaller values of $p$, $\R^{(2)}$ alone allows for detection of that Schmidt number), and ends at $p_\text{max}$, where the criterion fails to detect higher Schmidt numbers.}
    \label{fig:measurement_sim}
\end{figure*}

To conclude, let us estimate the total number of required measurements for detecting a violation of the criterion in Lemma~\ref{lem:dVschmidt} for the family of isotropic states,
\begin{align}\label{eq:iso}
    \rho_\text{iso}(p) = (1-p)\ketbra{\phi^+}{\phi^+} + \frac{p}{d_1d_2}\one,
\end{align}
parameterized by a noise parameter $p$. 
To that end, we simulate the measurement process of the two moments $\R^{(2)}$ and $\R^{(4)}$ for the observable $\mathcal{A}_d^{(4)}$ given in Eq.~\ref{eq:Ad4}. The total number of measurements is given by $MK$, where $M$ denotes the number of random local unitaries, whereas $K$ denotes the number of measurements per unitary in order to estimate the trace. 

We first simulate the measurement processes for fixed $MK$ while varying $M$. Looking at the achievable standard deviations of the two moments suggests that always choosing $K\approx 100$ will yield standard deviations close to the minimum. We therefore fix $K=100$ and simulate the measurement process $100$ times for each choice of noise parameter $p$ and dimension $d_1 = d_2 = d \in \{2,3,4,5\}$, and values $M\in\{[10, 30, 100, 300, 1000\}$. As expected, the standard deviation of the criterion over the 100 repetitions scales with $1/\sqrt{M}$, which allows us to extrapolate the deviation until we observe a $3\sigma$ violation of our criterion for different values of $k$. 

We display the estimated number of measurements $MK$ to obtain a $3\sigma$ violation of the criterion as a function of the noise parameter $p$ in Fig.~\ref{fig:measurement_sim}. 
The results suggest that entanglement detection (i.e., detecting Schmidt umber $k \geq 2$) in the regime where our criterion applies is possible starting from roughly $2\cdot10^4$ measurements for lower values of the parameter $p$, regardless of the dimensions, 
whereas detection of higher Schmidt numbers requires at least $10^5$ (for $k\geq 3$) or $10^6$ (for $k\geq 4$) measurements. Note that for increasing values of the noise parameter $p$ also the number of required measurements increases. 

\section{Conclusions}\label{sec:conclusion}

In this paper, we generalized entanglement criteria based on norms of correlation matrices of bipartite quantum states in order to detect their Schmidt numbers and connected them to moments of randomized measurements. We constructed observables whose random moments coincide with orthogonal moments for $t\leq d$, such that their measurement gives direct insight into the singular values of the correlation matrix of a quantum state. Furthermore, we provided explicit low-rank observables that yield the second and fourth moments, i.e., knowledge of the sum of squares and the sum of quarts of the singular values of the underlying state.

Further on, we explored the landscape of the second and fourth moments for arbitrary dimension $d$. In particular, we showed how to express the introduced Schmidt number criterion in this language and we showed that at least in the case of $d=2$ (and in the limit of large dimensions, provided that Zauner's conjecture is true), the subset of separable states predicted by the de~Vicente criterion is optimal. To do so, we constructed separable states for each point in the set resulting from the de~Vicente criterion by exploiting the existence of SIC-POVMs and sets of mutually unbiased bases in specific dimensions. In this way, we shed light onto the question of whether a given vector of singular values exists in a (separable) quantum state.

An immediate future research direction in this context would be a similar analysis for the subset of higher Schmidt numbers: Here, the question of whether our generalized criterion is optimal is even less clear. Nevertheless, the additional knowledge of higher moments $\Rt^{(t)}$ for $t>4$ allows to reconstruct the singular values of the correlation matrix more precisely and would eventually allow to detect all states that violate our criterion.

In principle, it should be even possible to detect all states of a certain Schmidt number using randomized measurements, as the tool allows one to measure all polynomial local unitary invariants \cite{wyderka2022complete}. Identifying the most relevant ones for bounding the Schmidt number of a state is a challenging task for further studies.

Given the connection between the existence of MUBs and SIC-POVMs to the present geometrical problem of characterizing the set of separable quantum states, it seems fruitful to fully understand the geometry of sets of singular values existing in states of limited Schmidt numbers. However, solving this task could be as demanding as solving the aforementioned existence problems directly.

Our findings pose the more general question of generalizations of known entanglement criteria to the detection of entanglement dimensionality. It has been known before that the reduction map criterion can be generalized to detect the Schmidt number \cite{terhal2000schmidt}, and recently, it was shown that also the covariance matrix criterion can be extended in this direction \cite{liu2022bounding}. The same question could be investigated for other criteria as well.

Finally, we assessed the experimental feasibility of the randomized measurement scheme to detect Schmidt numbers using our criterion and found a reasonable number of total measurements of around $10^6$ would suffice to detect Schmidt numbers of at least $2$, $3$ and $4$. This implies that the measurement scheme could be implemented in optical setups as well as relatively slow trapped ion experiments, where preparation and entangling gate times in the range of microseconds to milliseconds are usually limiting the repetition rate \cite{bruzewicz2019trapped}. In all of these setups, however, the implementation of Haar random local unitary gates could be challenging, and a thorough analysis of the impacts of deviating distributions on the results of randomized measurements would be advantageous for real world implementations.

\begin{acknowledgments}
We thank Felix Huber for fruitful discussion. N.W.~acknowledges support by the QuantERA project QuICHE via the German Ministry of Education and Research (BMBF Grant No.~16KIS1119K). A.K.~acknowledges funding from the  
Ministry of Economic
Affairs, Labour and Tourism Baden-Württemberg, under the
project QORA.

{\it Note:} During the preparation of this manuscript, we noticed that related results were derived in Ref.~\cite{liu2022characterizing}. While they derive the same criterion as our Lemma~\ref{lem:dVschmidt}, their proof is completely different and makes use of the covariance matrix criterion. Furthermore, they proceed by comparing the criterion to fidelity
based Schmidt rank witnesses and discuss generalizations to multi-partite states, whereas we focus on the optimality and experimental implementability.
\end{acknowledgments}

\appendix
\section{Proof of Lemma~\ref{lem:dVschmidt} and generalizations}\label{app:dVschmidt}
First, we prove the statement of Lemma~\ref{lem:dVschmidt} in the main text. After that, we generalize the statement to the whole family of entanglement criteria presented in \cite{Sarbicki-2020-Afamilyofmultipar}.
\setcounter{thmc}{0}

\begin{lemma}
Let $\rho_k$ be a bipartite quantum state as before with Bloch decomposition as in Eq.~(\ref{eq:blochdeco}) and $\text{SN}(\rho_k) \leq k$. Then
\begin{align}
    \Vert T \Vert_{\text{Tr}} \leq \sqrt{(d_1-1)(d_2-1)} + \sqrt{d_1d_2}(k-1).
\end{align}
\end{lemma}\begin{proof}
Let $\rho_k = \sum_i p_i \ketbra{\psi_i}{\psi_i}$ with $\text{SR}(\ket{\psi_i}) \leq k$. Then, $\Vert T \Vert_{\text{Tr}} \leq \sum_i p_i \Vert T(\ketbra{\psi_i}{\psi_i}) \Vert_{\text{Tr}}$.
For some fixed $i$, we write explicitly
\begin{align}
    \ketbra{\psi_i}{\psi_i} = \sum_{m,n=1}^{r} \sqrt{q_mq_n} \ketbra{a_m}{a_n} \otimes \ketbra{b_m}{b_n},
\end{align}
where $r \leq k$.
Each of the objects $\ketbra{a_m}{a_n}$ and $\ketbra{b_m}{b_n}$ can be decomposed in the operator bases $\{\lambda_i\}$, $\{\tilde{\lambda}_i\}$, respectively, such that
\begin{align}
    \ketbra{a_m}{a_n} = \frac{1}{d_1}[\delta_{mn}\one + \sum_{l=1}^{d_1^2-1} \alpha^{(m,n)}_l \lambda_l],\\
    \ketbra{b_m}{b_n} = \frac{1}{d_2}[\delta_{mn}\one + \sum_{l=1}^{d_2^2-1} \beta^{(m,n)}_l \tilde{\lambda_l}].
\end{align}
Note that the fact that the $\ket{a_m}$ and $\ket{b_m}$ are orthonormal fixes the constant in front of the identity. 

Using this decomposition, we get
\begin{align}
\Vert T(\ketbra{\psi_i}{\psi_i}) \Vert_{\text{Tr}} = \Vert \sum_{m,n=1}^r \sqrt{q_mq_n} \vec{\alpha}^{(m,n)} \vec{\beta}^{(m,n)\dagger}\Vert_{\text{Tr}} \nonumber \\
\leq \sum_{m,n=1}^r \sqrt{q_m q_n} \Vert \vec{\alpha}^{(m,n)} \Vert \Vert \vec{\beta}^{(m,n)} \Vert,
\end{align}
where we exploit the usual properties of the trace norm. Finally, we upper bound the Bloch vectors $\vec{\alpha}^{(m,n)}$ and $\vec{\beta}^{(m,n)}$. To that end, we write
\begin{align}
    1 = \trace(\ketbra{a_m}{a_n}\ketbra{a_n}{a_m}) = \frac{1}{d_1}(\delta_{mn} + \Vert \vec{\alpha}^{(m,n)} \Vert^2),
\end{align}
and likewise for $\beta^{(m,n)}$. Therefore,
\begin{align}
    \Vert T(\ketbra{\psi_i}{\psi_i}) \Vert_{\text{Tr}} &\leq \sum_{m,n=1}^r \sqrt{q_m q_n}\sqrt{d_1 - \delta_{mn}}\sqrt{d_2 - \delta_{mn}} \nonumber\\
    &=\sum_{m=1}^r q_m\sqrt{d_1 - 1}\sqrt{d_2-1} +\nonumber \\
    &\phantom{=======}+ 2\sum_{m<n=1}^r\sqrt{q_m q_n}\sqrt{d_1d_2}\nonumber \\
    &\leq \sqrt{d_1 - 1}\sqrt{d_2-1} + \sqrt{d_1d_2}(r-1),
\end{align}
as the maximal value of the sum is attained when $q_m = q_n = \frac1r$. Finally, using $r\leq k$ we obtain the claim.
\end{proof}

\setcounter{thmc}{5}
\begin{corollary}
    Let $\rho_k$ be a bipartite quantum state as before with Bloch decomposition as in Eq.~(\ref{eq:blochdeco}) and $\text{SN}(\rho_k) \leq k$. Let $x,y\geq 0$. Then
\begin{multline}
    \Vert D_x C D_y\Vert_{\text{Tr}} \leq \sqrt{(d_1-1+x^2)(d_2-1+y^2)} + \\ 
    + \sqrt{d_1d_2}(k-1),
\end{multline}
where $D_x = \operatorname{diag}(x,1,\ldots,1)$, $D_y = \operatorname{diag}(y,1,\ldots,1)$ and $C = \begin{pmatrix} 1 & \vec{\beta}^\dagger \\
\vec{\alpha} & T\end{pmatrix}$.
\end{corollary}
\begin{proof}
The proof follows the same lines as the proof of Lemma~\ref{lem:dVschmidt}. Instead of considering the correlation matrix $T$, we calculate the trace norm of $D_xCD_y$ as defined above. In this case, however,
\begin{multline}
\Vert D_xC(\ketbra{\psi_i}{\psi_i})D_y \Vert_{\text{Tr}} \leq \nonumber \\ 
\sum_{m,n=1}^r \sqrt{q_m q_n}  \sqrt{x^2\delta_{mn}+\Vert\vec{\alpha}^{(m,n)} \Vert^2} \sqrt{y^2\delta_{mn}+\Vert \vec{\beta}^{(m,n)} \Vert^2}.
\end{multline}
The rest of the proof remains the same, yielding the claim.
\end{proof}

\section{Equivalence of orthogonal and unitary moments}\label{app:suso}

In this appendix, we are going to construct observables $\mathcal{A}_d$, such that the unitary moments $\mathcal{R}^{(t)}_{\mathcal{A}_d}(\rho)$
coincide with the orthogonal moments $\mathcal{S}^{(t)}(\rho)$
for all $n$-partite quantum states of local dimension $d$ and for as many moments $t$ as possible. The outline of the argument is as follows. First, we describe the $n$-partite case in general, then we show that for $\mathcal{A}_d$ to yield the correct moments in this case, it is necessary that it does so already on a single-particle level. Then we show that for single particles, the spectrum of $\mathcal{A}_d$ is fixed already by the constraints that it should yield same moments for all odd $t$, as well as the even moments $t\leq d$. Next, we show that this fixed spectrum yields these coinciding moments also for $n$-partite states. Finally, we construct a specific simple rank $4$ observable that leads to coinciding second and fourth moment.

\subsection{The moments for $n$-partite states}

For a multipartite quantum state $\rho$ of local dimension $d$, we choose a local, hermitian basis $\{\lambda_0,\ldots,\lambda_{d^2-1}\}$ with $\lambda_0 = \one$, $\trace(\lambda_i\lambda_j) = d\delta_{ij}$ to write
\begin{align}
    \rho = \frac1{d^n} \sum_{i_1, \ldots, i_n=0}^{d^2-1} \alpha_{i_1\ldots i_n} \lambda_{i_1} \otimes \ldots \otimes \lambda_{i_n}.
\end{align}
To start, we define two different moments. First, we choose a local observable $A$ and write the moments of the distribution of measurement results under local unitary rotations, which map
\begin{align}
    \rho\rightarrow \rho_U = U_1\otimes \ldots \otimes U_n \rho U_1^\dagger \otimes \ldots \otimes U_n^\dagger,
\end{align}
as
\begin{align}
    \mathcal{R}_{\mathcal{A}_d}^{(t)}(\rho) &= \int \text{d}U_1\ldots \text{d}U_n \langle A\otimes \ldots \otimes A\rangle_{\rho_U}^t\\
    &= \int \text{d}U_1\ldots \text{d}U_n \langle U_1^\dagger \mathcal{A}_d U_1 \otimes \ldots \otimes U_n^\dagger \mathcal{A}_d U_n\rangle_\rho^t.
\end{align}
Second, we instead rotate the Bloch vectors locally, i.e.,
\begin{align}
    \rho \rightarrow \rho_{O} &= \frac1{d^n}\!\!\!\!\sum_{\substack{i_1,\ldots,i_n\\j_1,\ldots,j_n}}  (O_1)_{i_1,j_1}\ldots(O_n)_{i_n,j_n}\alpha_{j_1\ldots j_n}\lambda_{i_1}\otimes \lambda_{i_n}\\
    &=\frac1{d^n}\!\!\!\!\sum_{i_1,\ldots,i_n} \!\!\alpha_{i_1\ldots i_n} [O_1^T\vec{\lambda}]_{i_1}\otimes \ldots \otimes [O_n^T\vec{\lambda}]_{i_n},
\end{align}
where the $O_i$ are orthogonal matrices of size $d^2$, but the upper left entry is chosen to be $1$ in order to preserve the property of the basis that $\lambda_0 = \one$. Thus, effectively,
\begin{align}
O_i = \begin{pmatrix}1 & \vec{0}^T\\ \vec{0} & \tilde{O}_i
\end{pmatrix}
\end{align}
with $\tilde{O}_i \in O(d^2-1)$. Then, we choose arbitrarily one of the traceless basis elements to measure, e.g., $\lambda_1$, and define the moments w.r.t.~random local Bloch bases as 
\begin{align}
    \mathcal{S}^{(t)}(\rho) &= \int \text{d}\tilde{O}_1\ldots \text{d}\tilde{O}_n\langle \lambda_1 \otimes \ldots \otimes \lambda_1\rangle_{\rho_O}^t\\
    &=\int \text{d}\tilde{O}_1\ldots \text{d}\tilde{O}_n \langle (O_1\vec{\lambda})_1 \otimes \ldots \otimes (O_n\vec{\lambda})_1\rangle_\rho^t\\
    &=\frac{1}{V^n}\int \text{d}\vec{\alpha}_1\ldots \text{d}\vec{\alpha}_n \langle \vec{\alpha_1}\cdot\tilde{\lambda} \otimes \ldots \otimes \vec{\alpha}_n\cdot\tilde{\lambda}\rangle_\rho^t,
\end{align}
where the integration spans over all $d^2-1$-dimensional unit vectors $\vec{\alpha_i}$, $\tilde{\lambda}=(\lambda_1, \ldots, \lambda_{d^2-1})^T$, and $V=\frac{2\sqrt{\pi}^{d^2-1}}{\Gamma[(d^2-1)/2]}$ is the surface of the unit sphere in $d^2-1$ dimensions.

While the orthogonal moments $\mathcal{S}^{(t)}$ do not depend on the specific choice of observable, the unitary moments $\mathcal{R}_\mathcal{A}^{(t)}$ do depend on $\mathcal{A}$. We aim to answer the question: for a fixed local dimension $d$, can one choose $\mathcal{A}$ such that $\mathcal{S}^{(t)}(\rho) = \mathcal{R}_A^{(t)}(\rho)$ for all states $\rho$?

\subsection{Reduction to single qudits}

Before answering the question for $n$-partite states, we consider the following necessary condition: If $\mathcal{S}^{(t)}(\rho) = \mathcal{R}_\mathcal{A}^{(t)}(\rho)$ for all $\rho$, then it will be true for $\rho_1^{\otimes n}$, where $\rho_1$ is some single qudit state, implying $\mathcal{S}^{(t)}(\rho_1) = \mathcal{R}_\mathcal{A}^{(t)}(\rho_1)$ on the single particle level.

In order to find conditions for $\mathcal{A}$ in this case, we first evaluate the orthogonal moments for a general single-qudit state 
\begin{align}
    \rho = \frac1d(\one + \sum_{i=1}^{d^2-1} \alpha_i \lambda_i).
\end{align}
We get
\begin{align}
    \mathcal{S}^{(t)}(\rho)&=\frac1V\int\text{d}\vec{\beta} \langle\vec{\beta}\cdot\tilde{\lambda}\rangle_\rho^{t}\\
    &=\frac1V\int\text{d}\vec{\beta}\left(\sum_{i=1}^{d^2-1} \alpha_i\beta_i\right)^t.
\end{align}
We choose $d^2-1$-dimensional spherical coordinates and set w.l.o.g. $\vec{\alpha} =(0,\ldots,0,a)^T$, such that $\vec{\alpha}\cdot\vec{\beta} = \vert \alpha\vert \cos(\phi_1)$. Thus,
\begin{align}
    \mathcal{S}^{(t)}(\rho)&=\frac{\vert \alpha\vert^t}V\int_0^\pi\text{d}\phi_1\ldots\int_0^\pi\text{d}\phi_{d^2-3}\int_0^{2\pi}\text{d}\phi_{d^2-2} \times \nonumber \\
                           & \phantom{eee}\times \cos(\phi_1)^t\sin(\phi_1)^{d^2-3}\prod_{k=2}^{d^2-3}\sin(\phi_{d^2-1-k})^{k-1} \nonumber\\
    &=\frac{2\pi\vert \alpha\vert^t}{V}\delta_{t,\text{even}}\frac{\Gamma(\frac{d^2-2}{2})\Gamma(\frac{t+1}{2})}{\Gamma(\frac{d^2-1+t}{2})}\prod_{k=2}^{d^2-3}\left(\frac{\sqrt{\pi}\Gamma(\frac k2)}{\Gamma(\frac{k+1}2)}\right)\\
    &=\frac{2\sqrt{\pi}^{d^2-2}\vert \alpha\vert^t}{V}\delta_{t,\text{even}}\frac{\Gamma(\frac{t+1}2)}{\Gamma(\frac{d^2-1+t}2)}\\
    &= \frac{\vert \alpha\vert^t\Gamma(\frac{d^2-1}{2})\Gamma(\frac{t+1}2)}{\sqrt{\pi}\Gamma(\frac{d^2-1+t}{2})}\delta_{t,\text{even}}.
\end{align}
It will later become useful to express the Bloch length $\vert \alpha\vert^t$ in terms of the purity of $\rho$ via $\vert \alpha\vert^2 = d\trace(\rho^2) -1$ for the case of even $t$:
\begin{align}
    \mathcal{S}^{(t)}(\rho) &= \frac{\Gamma(\frac{d^2-1}{2})\Gamma(\frac{t+1}2)}{\sqrt{\pi}\Gamma(\frac{d^2-1+t}{2})}(d\trace(\rho^2)-1)^{t/2}\\
    &=\frac{\Gamma(\frac{d^2-1}{2})\Gamma(\frac{t+1}2)}{\sqrt{\pi}\Gamma(\frac{d^2-1+t}{2})}\!\sum_{j=0}^{t/2} \!\binom{t/2}{j}(-1)^{t/2-j} d^j \trace(\rho^2)^j.\label{eq:Smexpanded}
\end{align}
Now, we turn to the evaluation of $\mathcal{R}_\mathcal{A}^{(t)}(\rho)$. To that end, we first write
\begin{align}
    \mathcal{R}_\mathcal{A}^{(t)}(\rho) &= \int \text{d}U \trace(U\rho U^\dagger \mathcal{A})^t\\
    &= \trace\left(\rho^{\otimes t} \int \text{d}U (U^\dagger \mathcal{A} U)^{\otimes t}\right) \\
    &= \trace(\rho^{\otimes t} \tilde{A}^{(t)}),
\end{align}
where $\tilde{A}^{(t)}$ is the $t$-twirled operator $A$. As it is invariant under unitaries $U^{\otimes t}$, using Schur-Weyl duality it can be written as
\begin{align}
    \tilde{A}^{(t)} = \sum_{\pi \in S_t} a_\pi^{(t)} V_\pi,
\end{align}
where $S_t$ denotes the set of permutations of size $t$, and $V_\pi$ is its unitary representation, i.e., 
\begin{align}
    V_\pi = \sum_{i_1\ldots i_t=0}^{d-1} \ketbra{i_{\pi_1}\ldots i_{\pi_t}}{i_1\ldots i_t}.
\end{align}
The coefficients $a_\pi^{(t)}$ are rather restricted. First, note that
\begin{align}
    \trace(\rho^{\otimes t} \tilde{A}^{(t)})=\sum_\pi a_\pi^{(t)} \trace(V_\pi \rho^{\otimes t}).
\end{align}
The trace with the product of the $\rho$'s will produce products of $\trace(\rho^k)$ according to the cycle type of $\pi$. For example, if $t=4$ and $\pi = (12) \equiv (12)(3)(4)$, meaning that $1$ is mapped to $2$, $2$ is mapped to $1$ and $3$ and $4$ remain unchanged, then $\trace(V_\pi \rho^{\otimes 4}) = \trace(\rho^2)$. As we will set this equal to Eq.~(\ref{eq:Smexpanded}), which contains only terms of $\trace(\rho)$ and $\trace(\rho^2)$, $a_\pi^{(t)} = 0$ whenever $\pi$ contains a cycle of length $3$ or longer. We therefore can write
\begin{align}
    \tilde{A}^{(t)} = a_{()}^{(t)}\one + \sum_{i<j}a_{(ij)}^{(t)}V_{(ij)} + \!\!\!\!\!\!\!\!\sum_{\substack{i<j, k<l, i<k,\\ j\neq l, j \neq k}}\!\!\!\!\!\!\!\!a_{(ij)(kl)}^{(t)} V_{(ij)(kl)} + \ldots.
\end{align}
As $\tilde{A}^{(t)}$ is invariant under exchange of any of the $t$ systems, we get $a_{(ij)}^{(t)} = a_{(kl)}^{(t)}$, etc. Thus, the coefficients only depend on the cycle type, and, as there are only cycles of length 2 (and identities), only on the number of two-cycles:
\begin{align}
    \tilde{A}^{(t)} &= a_{0}^{(t)}\one + a_{1} \sum_{i<j}V_{(ij)} + a_{2}^{(t)}\!\!\!\!\!\!\sum_{\substack{i<j, k<l, i<k,\\ j\neq l, j \neq k}} V_{(ij)(kl)} + \ldots\\
    &=a_{0}^{(t)}\one + \sum_{k=1}^{t/2}a_k^{(t)}\!\!\!\!\!\!\sum_{(i_1,j_1),\ldots,(i_k,j_k)}\!\!\!\!\!\!V_{(i_1,j_1)\ldots(i_k,j_k)},\label{eq:Atdecomp}
\end{align}
where the sum over the two-cycles only runs over proper two-cycles, i.e., no index occurs twice and $i_l < j_l$ for all $l$. In this sum, there are $\binom{t}{2k}(2k-1)!!$ terms.

\subsection{Evaluating the $a_k^{(t)}$}

Let us now try to fix the coefficients $a_k^{(t)}$. To that end, we evaluate using Eq.~(\ref{eq:Atdecomp})
\begin{align}\label{eq:ratintrrho2}
    \mathcal{R}_A^{(t)} &= \trace(\tilde{A}^{(t)} \rho^{\otimes t}) \\
    &= a_0^{(t)}+\sum_{k=1}^{t/2}a_k^{(t)}\binom{t}{2k}(2k-1)!!\trace(\rho^2)^k,
\end{align}
which we set equal to Eq.~(\ref{eq:Smexpanded}) and compare the coefficients. We get
\begin{align}
    a_k^{(t)} &= \frac{(-1)^{t/2-k}d^k\Gamma(\frac{d^2-1}2)(t-1)!!\binom{t/2}{k}}{2^{t/2}\Gamma(\frac{d^2-1+t}2)(2k-1)!!\binom{t}{2k}} \\
    &= (-1)^{t/2-k}d^k\frac{(d^2-3)!!(t-2k-1)!!}{(d^2-3+t)!!}\label{eq:ak}
\end{align}
for $k=0,\ldots,t/2$ (we use the convention $(-1)!! = 1$).
Thus, the operator $\tilde{A}^{(t)}$ is completely fixed, and all that remains is to translate it back into the observable $A$.

\subsection{Evaluating $\trace(\mathcal{A}^t)$}

After all of these preparations, we can now tackle the task of obtaining $\mathcal{A}$. Of course, $\mathcal{A}$ is only fixed up to unitary rotations, which means that we can only hope to fix its eigenvalues. Equivalently, we will fix $\trace(\mathcal{A}^t)$ for all $t$, from which the eigenvalues can be recovered.
Let us start with odd powers, and restart from the expression of the unitary moment, but this time we twirl the state $\rho$, yielding $\tilde{\rho}^{(t)} = \int \text{d}U (U\rho U^\dagger)^{\otimes t} = \sum_{\pi \in S_t} r_\pi V_\pi$. This yields
\begin{align}
    \mathcal{R}_\mathcal{A}^{(t)} = \sum_{\pi \in S_t} r_\pi \trace(V_\pi \mathcal{A}^{\otimes t}),
\end{align}
where $\trace(V_\pi \mathcal{A}^{\otimes t})$ is again a product of traces of powers of $\mathcal{A}$. As $t$ is odd, each of these terms has to contain the trace of an odd power of $\mathcal{A}$, and the permutations corresponding to single cycles of length $t$ all yield $\trace(\mathcal{A}^t)$.

Now, starting with $t=1$, we get directly $0 = \mathcal{R}_\mathcal{A}^{(1)}(\rho) = \trace(\mathcal{A})$.
From $t=3$, we get expressions with $\trace(\mathcal{A})$, which vanish, and a contribution of $\trace(\mathcal{A}^3)$, which therefore has to vanish as well. By induction, we conclude
\begin{align}
    \trace(\mathcal{A}^t) = 0
\end{align}
for $t$ odd. This implies that the eigenvalues of $\mathcal{A}$ must come in pairs $\pm x$.

Now, let us turn to the case of even $t$. In this case, we evaluate
\begin{align}\label{eq:tratfromattv}
    \trace(\tilde{A}^{(t)}V_{(1,2\ldots t)}) &= \int \text{d}U \trace[ (U^\dagger \mathcal{A} U)^{\otimes t} V_{(1,2\ldots t)}] \\
    &= \int \text{d}U \trace[(U^\dagger \mathcal{A} U)^t] = \trace(\mathcal{A}^t).
\end{align}
On the other hand, inserting Eq.~(\ref{eq:Atdecomp}) and using the fact that $\trace(V_\pi) = d^{\# \text{cycles}(\pi)}$, yields
\begin{multline}\label{eq:trattheo}
    \trace(\tilde{A}^{(t)}V_{(1,2\ldots t)}) =\\
    a_0^{(t)} d + \sum_{k=1}^{t/2} a_k^{(t)} \!\!\!\!\!\!\sum_{(i_1,j_1)\ldots(i_k,j_k)}\!\!\!\!\!\! \trace(V_{(i_1,j_1)\ldots(i_k,j_k)}V_{(1,2,\ldots t)}).
\end{multline}

\begin{figure}[t]
    \centering
    \includegraphics[width=1.0\columnwidth]{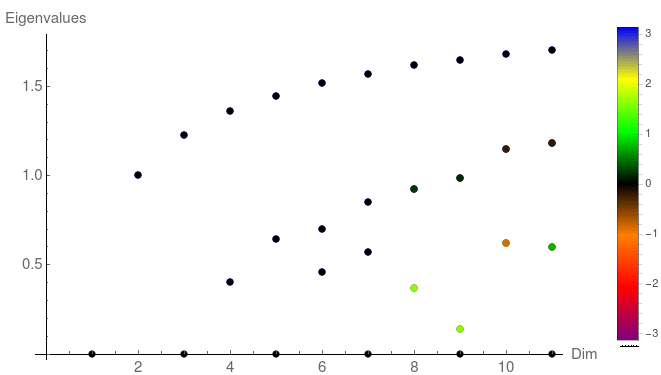}
    \caption{The absolute values of the eigenvalues of the observable $A$ that reproduces orthogonal moments up to order $d$ as a function of the dimension $d$. Note that starting from $d=8$, complex valued eigenvalues are needed. Their complex phases are indicated by the color.}
    \label{fig:tall_evals}
\end{figure}
Now, the difficulty lies in evaluating the trace on the right hand side. In fact, we find
\begin{multline}
    \sum_{(i_1,j_1)\ldots(i_k,j_k)}\!\!\!\!\!\!\trace(V_{(i_1,j_1)\ldots(i_k,j_k)}V_{(1,2,\ldots t)}) = \\
    \binom{t}{2k}\sum_{g=0}^{\lfloor k/2\rfloor}a(k,g) d^{k+1-2g},
\end{multline}
where the numbers $a(n,g)$ denote the number of ways to glue sides of a $2n$-gon so as to produce a surface of genus $g$.
They are given in \url{https://oeis.org/A035309}. We can now insert the expressions for $a_k^{(t)}$ from Eq.~(\ref{eq:ak}) and explicitly evaluate for fixed dimension $d$ $\trace(\mathcal{A}^t)$ for $t\leq d$. Finally, from knowledge of the traces, we can uniquely infer the eigenvalues of $\mathcal{A}$.

\subsection{Examples}

We construct the observables for $d\leq 7$. For $d>8$, we need to allow for complex eigenvalues.

\begin{align}
    \mathcal{A}_2 &= \diag(1,-1),\\
    \mathcal{A}_3 &= \diag(\sqrt{3/2},0,-\sqrt{3/2}) \approx \diag(1.225, 0, -1.225),\\
    \mathcal{A}_4 &= \diag\left(\sqrt{1+2 \sqrt{3/17}},\sqrt{1-2 \sqrt{3/17}}, \right.\nonumber\\
      &\phantom{======}\left.-\sqrt{1-2 \sqrt{3/17}},-\sqrt{1+2 \sqrt{3/17}}\right) \nonumber \\
    &\approx \diag(1.357, 0.400, -0.400, -1.357),\\
    \mathcal{A}_5 &\approx \diag(1.444, 0.644, 0, -0.644, -1.444), \\
    \mathcal{A}_6 &\approx \diag(1.518, 0.695, 0.459, -0.459, -0.695, -1.518),\\
    \mathcal{A}_7 &\approx \diag(1.567, 0.851, 0.567, 0, -0.567, -0.851,  -1.567).
\end{align}
The solutions are displayed in Fig.~\ref{fig:tall_evals}.

While these operators yield matching moments for $t\leq d$, they fail to do so in general for larger $t$, as can be readily checked by comparing $\trace(\mathcal{A}_d^t)$ with Eq.~(\ref{eq:trattheo}) for some $t>d$. Only in case of $d=2$, they match for all $t$, and for $d=3$, they match up to $t=4$. Otherwise, they differ.

If we just want an observable that reproduces moments $t=1,2,3,4$ with as few non-vanishing eigenvalues as possible, we can construct them by setting equal Eqs.~(\ref{eq:ratintrrho2}) and (\ref{eq:Smexpanded}) just for $t\leq4$. This fixes the $a_k^{(t)}$ for these $t$, and via Eq.~(\ref{eq:tratfromattv}) also $\trace(\mathcal{A}^t)$ for $t \leq 4$. This yields an explicit solution for $d\geq 2$:

\begin{align}
    \mathcal{A}_d^{(4)} = \diag( \pm \kappa_1, \pm \kappa_2, 0, \ldots, 0)
\end{align}
with 
\begin{align}
    \kappa_{1,2} = \frac12 \sqrt{d \pm \sqrt{d\left(8-d-\frac{20}{d^2+1}\right)}}.
\end{align}

Note that these solutions in general can only reproduce moments up to order $t=d$, as the unique solutions for the eigenvalues fail to match $\trace(\mathcal{A}^t)$ for $t>d$. However, there are two exceptions. First, for $d=2$, all moments match, as $\trace(\mathcal{A}^t) = 2$ for all even $t$. Second, for $d=3$, also $t=4$ is magically reproduced.
The eigenvalues for these fourth-moment observables are displayed in Fig.~\ref{fig:t24_evals}.

\begin{figure}
    \centering
    \includegraphics[width=1.0\columnwidth]{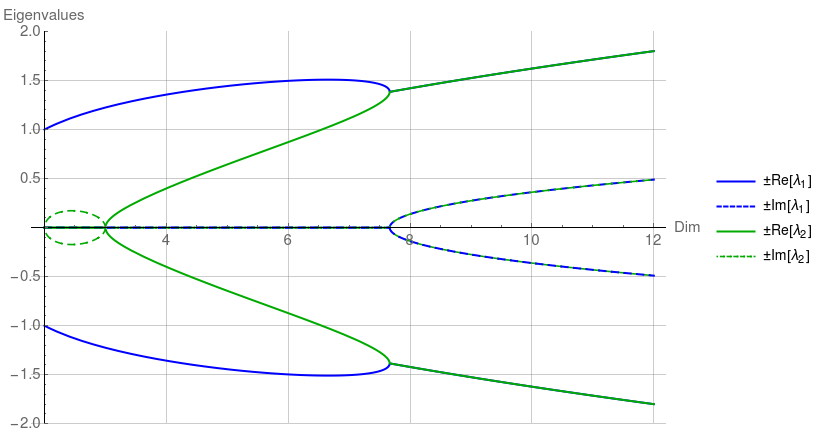}
    \caption{The (up to) four non-vanishing eigenvalues of the observable $A$ that reproduces orthogonal moments up to order four.}
    \label{fig:t24_evals}
\end{figure}

\subsection{Generalization to $n$-partite states}
The last missing step is to argue that the solutions obtained for single particles  also provide solutions to the $n$-partite case. To that end, notice that
\begin{align}\label{eq:Strhon}
\mathcal{S}^{(t)}(\rho_n) = \trace[\rho_n^{\otimes t} \tilde{S}^{(t) \otimes n}],
\end{align}
where $\tilde{S}^{(t)} = \intop \text{d}\vec{\alpha} (\vec{\alpha}\cdot \vec{\lambda})^{\otimes t}$
and, likewise
\begin{align}\label{eq:Rtrhon}
\mathcal{R}^{(t)}_{\mathcal{A}}(\rho_n) = \trace(\rho_n^{\otimes t} \tilde{A}^{(t)\otimes n}).
\end{align}
Now, thanks to the fact that $\text{SU}(d)$ is a proper subgroup of the special orthogonal group $\text{SO}(d^2-1)$, for every unitary rotation $U \in \text{SU}(d)$, there exists an orthogonal rotation $O \in \text{SO}(d^2-1)$, s.t. $U\vec{\alpha}\cdot\vec{\lambda}U^\dagger = (O\vec{\alpha})\cdot\vec{\lambda}$, implying that $\tilde{S}^{(t)}$ is invariant under local unitary transformations $U^{\otimes t}$, just like $\tilde{A}^{(t)}$. Using once more Schur-Weyl duality, this implies that $\tilde{S}^{(t)} = \sum_{\pi \in S_t} s_\pi V_\pi$, and comparing with Eq.~(\ref{eq:Smexpanded}) fixes $\tilde{S}^{(t)} = \tilde{A}^{(t)}$ for all $t$. Thus, Eqs.~(\ref{eq:Rtrhon}) and (\ref{eq:Strhon}) must coincide for all $n$, if they coincide for $n=1$.

\section{Optimization of the moment landscape}\label{app:opti}

In order to find optimized entanglement criteria using the second and fourth moments, we express them in terms of the local unitary invariants $\trace(TT^\dagger )$ and $\trace(TT^\dagger TT^\dagger )$, as
\begin{align}
    \Rt^{(2)} &= \frac{1}{(d_1-1)(d_2-1)}\trace(TT^\dagger ), \\
    \Rt^{(4)} &= \frac{1}{3(d_1-1)^2(d_2-1)^2}\left[2\trace(TT^\dagger TT^\dagger ) + \trace(TT^\dagger )^2\right].
\end{align}
Here, $T$ denotes the $(d_1^2-1) \times (d_2^2-1)$-dimensional correlation matrix of the quantum state as defined in the main text. 

We assume w.l.o.g.~that $d_1 \leq d_2$ and we denote the singular values of $T$ as $(\sigma_1,\ldots, \sigma_{d_1^2-1})$. Then, $\trace(TT^\dagger) = \sum_i \sigma_i^2$ and $\trace(TT^\dagger TT^\dagger) = \sum_i \sigma_i^4$. We now minimize and maximize $\trace(TT^\dagger TT^\dagger)$ for a fixed value of $\trace(TT^\dagger)$ under the constraint dictated by the generalized de~Vicente criterion for states of Schmidt number $k$ from Lemma~\ref{lem:dVschmidt}, namely $\sum_i \sigma_i \leq \sqrt{(d_1-1)(d_2-1)} + \sqrt{d_1d_2}(k-1)$, yielding the optimization problem
\opti{min/max}{\{\sigma_i\}}{\sum_i \sigma_i^4}{\sigma_i \geq 0,;\sum_i \sigma_i^2 = (d_1-1)(d_2-1)\Rt^{(2)},;\sum_i \sigma_i \leq \sqrt{(d_1-1)(d_2-1)} + \sqrt{d_1d_2}(k-1).}
The optimization can be solved analytically using the standard Lagrange multiplier method with slack variables to implement the upper bound. This readily reveals that the lower bound is given by arrangements of singular values that either vanish or are equal, except for one singular value that is allowed to take different values.

For the case of separable states, i.e., $k=1$, we find the analytical expression $f_\text{lb}(\Rt^{(2)}) \leq \Rt^{(4)} \leq f_\text{ub}(\Rt^{(2)})$ with
\begin{widetext}
\begin{align}\label{eq:flb}
      f_{\text{lb}} (\Rt^{(2)}) = \begin{cases}
        \frac{d_1^2+1}{3(d_1^2-1)}(\Rt^{(2)})^2 & 0\leq \Rt^{(2)} \leq \frac1{d^2-1} \\
        \frac13(\Rt^{(2)})^2+\frac{2}{3(d_1^2-m)^4} \left[ \left(\sqrt{v(\Rt^{(2)},m)} -1\right)^4 + \frac{\left(\sqrt{v(\Rt^{(2)},m)}+d_1^2-m-1\right)^4}{(d_1^2-m-1)^3}   \right] &
        \frac1{d_1^2-m} \leq \Rt^{(2)} \leq \frac1{d_1^2-m-1}, \\
        \vdots &  m=1,\ldots,d_1^2-2
      \end{cases}
\end{align}
\end{widetext}
where $v(\Rt^{(2)},m) = (d_1^2-m)(d_1^2-m-1)(\Rt^{(2)} - \frac1{d_1^2-m})$.
The upper bound is in the case of $k=1$ simply given by
\begin{align}\label{eq:fub}
    f_\text{ub}(\Rt^{(2)}) = (\Rt^{(2)})^2.
\end{align}
For $k\geq 2$, the maximal value of $\Rt^{(4)}$ depends on the maximal value that a single singular value can take. To the best of our knowledge, this upper bound is not known, and the best bound that we could find is dictated by the purity bound $\trace(\rho^2) \leq 1$, which translates to a bound on the Frobenius norm of $T$,
\begin{align}
    \Vert T \Vert_\text{Fr} = \sum_{i,j} \vert T_{ij} \vert^2 \leq d_1d_2 + 1 - \frac{d_1+d_2}k,
\end{align}
which can be derived by rewriting $\trace(\rho^2) = \frac1{d_1d_2}( d_1\trace(\rho_A^2) + d_2\trace(\rho_B^2) + \Vert T \Vert_\text{Fr} - 1)$ and maximizing over pure states of Schmidt rank $k$. The upper bound is then given by Eq.~(\ref{eq:fub}) as well, as long as $\Rt^{(2)}\leq 1+\frac{k-1}{k}\frac{d_1+d_2}{(d_1-1)(d_2-1)}$, as values larger than this are incompatible with Schmidt number $k$ states.
However, the upper bound is probably not tight, except for the case of $k=1$.

The resulting picture is displayed for the case of $d_1 = d_2 = 3$ in Fig.~\ref{fig:fig1}.

The lower bound for the whole set of states is traced by the isotropic state, and is given by
\begin{align}
    g_{\text{lb}}(\Rt^{(2)}) = \frac{d_1^2+1}{3(d_1^2-1)}(\Rt^{(2)})^2.
\end{align}

\section{Proof of Theorem~\ref{thm:Thm1}}\label{app:proofthm1}
Here, we prove Theorem~\ref{thm:Thm1} from the main text.
\setcounter{thmc}{2}

\begin{theorem}
Let $\{\psi_i\}_{i=1}^{N}$ be an equiangular set of normalized vectors in $\mathbb{C}^d$, with $\vert\braket{\psi_i|\psi_j}\vert = \frac{1}{\sqrt{d}}$ if $i\neq j$. Then
\begin{align}
    \rho = \frac1N\sum_{i=1}^N \ket{\psi_i}\!\bra{\psi_i} \otimes \ket{\psi_i}\!\bra{\psi_i}
\end{align}
lies on the kink at $\Rt^{(2)}(\rho) = \frac{1}{N}$.
\end{theorem}
\begin{proof}
As noted before, lying on the kink at $\Rt^{(2)} = \frac1N$ is equivalent to the singular values of the correlation matrix being $(0,\ldots,0,\frac{d-1}N,\ldots,\frac{d-1}N)$, summing up to $d-1$. The elements of the correlation matrix are given by
\begin{align}
    T_{ij} &= \trace(\lambda_i \otimes \lambda_j \rho) \\
           &= \frac1N \sum_{k=1}^{N} \braket{\psi_k|\lambda_i|\psi_k}\braket{\psi_k|\lambda_j|\psi_k} \\
           &= \vec{\kappa}_i \cdot \vec{\kappa}_j,
\end{align}
with the $N$-dimensional vectors $\vec{\kappa}_i = (\braket{\psi_1|\lambda_i|\psi_1}/\sqrt{N},\ldots,\braket{\psi_N|\lambda_i|\psi_N}/\sqrt{N})^T$.
This implies that $T$ is a $d^2-1$-dimensional Gramian matrix of rank $N$, implying that it is positive semidefinite and its eigenvalues coincide with the singular values. Thus, if we show that $T^2 = \frac{d-1}N T$ and $\trace(T)=d-1$, we are done.
We calculate
\begin{align}
    (T^2)_{ij} &= \sum_a T_{ia}T_{aj} \nonumber\\
                &= \frac1{N^2}\sum_{a,k,l} \braket{\psi_k|\lambda_i|\psi_k}\braket{\psi_k|\lambda_a|\psi_k}\braket{\psi_l|\lambda_a|\psi_l}\braket{\psi_l|\lambda_j|\psi_l} \nonumber\\
                & = \frac1{N^2}\sum_{k,l} \braket{\psi_k|\lambda_i|\psi_k}\braket{\psi_l|\lambda_j|\psi_l} \times\nonumber \\ &\phantom{===========}\times\bra{\psi_k}\!\bra{\psi_l}\sum_a \lambda_a \otimes \lambda_a \ket{\psi_k}\!\ket{\psi_l} \nonumber \\
                & = \frac1{N^2}\sum_{k,l} \braket{\psi_k|\lambda_i|\psi_k}\braket{\psi_l|\lambda_j|\psi_l} \times \nonumber \\
               &\phantom{==========}\times\bra{\psi_k}\!\bra{\psi_l}(dV-\one\otimes\one) \ket{\psi_k}\!\ket{\psi_l} \nonumber \\
                &= \frac{1}{N^2}\sum_{k,l} \braket{\psi_k|\lambda_i|\psi_k}\braket{\psi_l|\lambda_j|\psi_l}[(d-1)\delta_{kl} + 1 - 1] \nonumber \\
                &= \frac{d-1}{N}\vec{\kappa}_i \cdot \vec{\kappa}_j = \frac{d-1}{N} T_{ij},
\end{align}
where $V$ denotes the swap operator, given by $V = \sum_{i,j=0}^{d-1} \ketbra{ij}{ji}$, and we use the well known relation $dV = \sum_{a=0}^{d^2-1} \lambda_a \otimes \lambda_a$. The trace can be evaluated analogously.
\end{proof}

\section{Proof of Corollary~\ref{cor:cor1}}\label{app:proofcor1}

Here, we prove Corollary~\ref{cor:cor1} from the main text.

\setcounter{thmc}{3}

\begin{corollary}
Let $\{\psi_i\}_{i=1}^{N} \cup \{\ket{\psi}\}$ be an equiangular set of size $N+1$ of normalized vectors in $\mathbb{C}^d$ with overlap $\frac{1}{\sqrt{d}}$. Then the family of states
\begin{align}
    \rho(p) = \frac{p}N\sum_{i=1}^N \ket{\psi_i}\!\bra{\psi_i} \otimes \ket{\psi_i}\!\bra{\psi_i} + (1\!-\!p) \ket{\psi}\!\bra{\psi} \otimes \ket{\psi}\!\bra{\psi}
\end{align}
traces the boundary of the separable set between the kinks at $\Rt^{(2)}[\rho(N/(N+1))] = \frac{1}{N+1}$ and $\Rt^{(2)}[\rho(1)] = \frac{1}{N}$.
\end{corollary}
\begin{proof}
The proof follows similar lines as the proof of Theorem~\ref{thm:Thm1}. We can write the correlation matrix explicitly as 
\begin{multline}
    T_{ij} = \frac{p}{N}\sum_k \braket{\psi_k|\lambda_i|\psi_k}\braket{\psi_k|\lambda_j|\psi_k} + \\
    + (1-p)\braket{\psi|\lambda_i|\psi}\braket{\psi|\lambda_j|\psi}, 
\end{multline}
which can, again, be written as an inner product, implying that $T$ is a Gramian matrix. As noted before, lying on the boundary of the set of separable states between the kinks at $\Rt^{(2)} = \frac1N$ and $\Rt^{(2)} = \frac1{N+1}$ is equivalent to the singular values of $T$ being given by $(0,\ldots,0,l,m,\ldots,m)$ with $0 \leq l \leq m$ and the singular value $m$ appearing $N$ times. Furthermore, the sum of the singular values must equal $d-1$. 

As the eigenvalues of the Gramian matrix $T$ coincide with its singular values, the claim is therefore equivalent to showing
\begin{align}\label{eq:T2}
    T^2 = mT + (l^2-ml) \ket{v_l}\!\bra{v_l},
\end{align}
where $\ket{v_l}$ is the eigenvector of $T$ with eigenvalue $l$. Similar calculations as before and making use of $\sum_{a=1}^{d^2-1}\lambda_a\otimes \lambda_a = dV-\one$ reveal
\begin{multline}
    (T^2)_{ij} = \sum_{a} T_{ia} T_{aj} = \frac{(d-1)p}{N}T_{ij} + \\
    + (d-1)^2[(1-p)^2 - \frac{p(1-p)}{N}]\braket{\psi|\lambda_i|\psi}\braket{\psi|\lambda_j|\psi}. \label{eq:T2rhs}
\end{multline}
Indeed, the last term in Eq.~(\ref{eq:T2rhs}) can be interpreted as an element of a rank one matrix. Comparing Eq.~(\ref{eq:T2rhs}) with Eq.~(\ref{eq:T2}), we read off
\begin{align}
    m &= \frac{(d-1)p}{N},\\
    l &= (d-1)(1-p), \\
    \ket{v_l} &= \frac{1}{\sqrt{d-1}}(\braket{\psi|\lambda_1|\psi}, \ldots, \braket{\psi|\lambda_{d^2-1}|\psi})^T.
\end{align}
Finally, we have to check two things. First, the sum of the singular values, $l+Nm$, must sum to $d-1$, which is easily confirmed. Second, we have to ensure that $\ket{v_l}$ is the eigenvector of $T$ with eigenvalue $l$. To that end, we insert $T_{ij}$ into
\begin{align}
    \braket{v_l|T|v_l} &= \frac{1}{d-1}\sum_{ij}\braket{\psi|\lambda_i|\psi}\braket{\psi|\lambda_j|\psi}T_{ij} \\
    &= (d-1)(1-p) = l. 
\end{align}
This concludes the proof.
\end{proof}

\section{Proof of Theorem~\ref{thm:Thm2}}\label{app:proofthm2}

Here, we prove Theorem~\ref{thm:Thm2} from the main text.

\setcounter{thmc}{4}

\begin{theorem}
Consider a set of $m$ MUBs  $\{\ket{\psi_j^1}\}_{j=1}^d, \ldots, \{\ket{\psi_j^m}\}_{j=1}^d$. Then the state
\begin{align}
    \rho &:= \frac{1}{md}\sum_{l=1}^m\sum_{k=1}^{d}  \ket{\psi_k^l}\!\bra{\psi_k^l} \otimes  \ket{\psi_k^l}\!\bra{\psi_k^l}
\end{align}
lies at the kink at $\Rt^{(2)} = \frac1{m(d-1)}$.
\end{theorem}
\begin{proof}
We show again that the correlation matrix is a multiple of a projector, i.e., $T^2 = \frac{d-1}{K}T$. To that end, we note that $\vert \braket{\psi_k^l|\psi_{k^\prime}^{l^\prime}}\vert^2 = \delta_{ll^\prime}\delta_{kk^\prime} + (1-\delta_{ll^\prime})\frac1d$ and write explicitly 
\begin{align}
    T_{ij} = \frac1{md}\sum_{l=1}^m\sum_{k=1}^d \braket{\psi_k^l|\lambda_i|\psi_k^l}\braket{\psi_k^l|\lambda_j|\psi_k^l}
\end{align}
and
\begin{widetext}
\begin{align}
    T^2_{ij} &= \frac1{m^2d^2}\sum_{a=1}^{d^2-1}\sum_{l,l^\prime=1}^m\sum_{k,k^\prime=1}^d \braket{\psi_k^l|\lambda_i|\psi_k^l}\braket{\psi_{k^\prime}^{l^\prime}|\lambda_j|\psi_{k^\prime}^{l^\prime}}\braket{\psi_k^l|\lambda_a |\psi_k^l}\braket{\psi_{k^\prime}^{l^\prime}|\lambda_a|\psi_{k^\prime}^{l^\prime}} \\
             &= \frac1{m^2d^2}\sum_{l,l^\prime=1}^m\sum_{k,k^\prime=1}^d \braket{\psi_k^l|\lambda_i|\psi_k^l}\braket{\psi_{k^\prime}^{l^\prime}|\lambda_j|\psi_{k^\prime}^{l^\prime}}\bra{\psi_k^l}\braket{\psi_{k^\prime}^{l^\prime}|dV-\one|\psi_k^l}\ket{\psi_{k^\prime}^{l^\prime}} \\
             & = \frac1{m^2d^2}\sum_{l,l^\prime=1}^m\sum_{k,k^\prime=1}^d \braket{\psi_k^l|\lambda_i|\psi_k^l}\braket{\psi_{k^\prime}^{l^\prime}|\lambda_j|\psi_{k^\prime}^{l^\prime}}(d\delta_{ll^\prime}\delta_{kk^\prime} + (1-\delta_{ll^\prime}) - 1) \\
             & = \frac1{m^2d^2}\sum_{l=1}^m\sum_{k=1}^d [d\braket{\psi_k^l|\lambda_i|\psi_k^l}\braket{\psi_k^l|\lambda_j|\psi_k^l}-\sum_{k^\prime=1}^d \braket{\psi_k^l|\lambda_i|\psi_k^l}\braket{\psi_{k^\prime}^l|\lambda_j|\psi_{k^\prime}^l}] \\
             &=\frac{d}{m^2d^2}\sum_{l=1}^m\sum_{k=1}^d \braket{\psi_k^l|\lambda_i|\psi_k^l}\braket{\psi_k^l|\lambda_j|\psi_k^l} = \frac{1}{m}T_{ij}.
\end{align}
\end{widetext}
Rearranging $\frac1{m}=\frac{d-1}{K}$ yields $K=m(d-1)$.
\end{proof}

\bibliography{cite}

\end{document}